\documentclass[reprint,superscriptaddress,nofootinbib,nobibnotes,amsmath,amssymb,prl]{revtex4-2}

\usepackage{graphicx}
\usepackage{float}
\usepackage[T1]{fontenc}

\usepackage[
  colorlinks,
  linkcolor = blue,
  citecolor = blue,
  urlcolor = blue]{hyperref}

\usepackage{amsmath,amsthm,amssymb,amsfonts}

\usepackage{thmtools}

\usepackage{cleveref}%

\usepackage{mathtools}

\usepackage[shortlabels]{enumitem}

\usepackage{makecell}
\usepackage{multirow}
\usepackage[table,xcdraw]{xcolor}
\usepackage{hhline}
\usepackage{diagbox}

\usepackage[export]{adjustbox}
\usepackage{float}

\usepackage{dsfont}

\DeclarePairedDelimiter{\set}{\lbrace}{\rbrace}
\DeclarePairedDelimiter{\abs}{\lvert}{\rvert}
\DeclarePairedDelimiter{\norm}{\lVert}{\rVert}

\DeclarePairedDelimiter{\of}{\lparen}{\rparen}
\DeclarePairedDelimiter{\sof}{\lbrack}{\rbrack}

\newcommand{\defeq}{\vcentcolon=}

\renewcommand{\leq}{\leqslant}
\renewcommand{\geq}{\geqslant}

\newcommand{\bra}[1]{\langle{#1}\rvert}
\newcommand{\ket}[1]{\lvert{#1}\rangle}

\newcommand{\ketbra}[2]{\ket{#1}\bra{#2}}
\newcommand{\proj}[1]{\ketbra{#1}{#1}}

\newcommand{\ct}{^{\dagger}}
\newcommand{\tp}{^{\mathsf{T}}}

\newtheorem{definition}{Definition}
\newtheorem{theorem}[definition]{Theorem}
\newtheorem{corollary}[definition]{Corollary}
\newtheorem{lemma}[definition]{Lemma}

\newtheorem{remark}[definition]{Remark}
\newtheorem{example}[definition]{Example}

\newcommand{\dd}{\mathrm{d}}

\newcommand{\dket}[1]{\vert {#1} \rangle\!\rangle}

\newcommand{\dketbra}[1]{\vert {#1} \rangle\!\rangle\!\langle\!\langle {#1} \vert}
\newcommand{\End}{\mathrm{End}}
\newcommand{\Tr}{\mathrm{Tr}}
\newcommand{\C}{\mathbb{C}}
\newcommand{\U}{\mathrm{U}}
\newcommand{\USch}{U_{\mathrm{Sch}}}
\newcommand{\UPW}{U_{\mathrm{PW}}}
\newcommand{\bUSch}{\bar{U}_{\mathrm{Sch}}}
\newcommand{\UdSch}{U_{\mathrm{dSch}}}

\newcommand{\cO}{\mathrm{cO}}%
\newcommand{\fO}{\mathrm{fO}}%
\newcommand{\iO}{\mathrm{iO}}%
\newcommand{\tO}{\mathrm{tO}}%

\newcommand{\CG}{\mathrm{CG}}
\newcommand{\dCG}{\mathrm{dCG}}
\newcommand{\Sch}{\mathrm{Sch}}
\newcommand{\dSch}{\mathrm{dSch}}
\newcommand{\GT}{\mathrm{GT}}
\newcommand{\SYT}{\mathrm{SYT}}
\newcommand{\SSYT}{\mathrm{SSYT}}
\newcommand{\margx}{\mathrm{marg}_x}
\newcommand{\margy}{\mathrm{marg}_y}
\newcommand{\e}{\mathbf{e}}

\newcommand{\1}{\mathds{1}}

\crefname{figure}{Fig.}{Figs.}
\crefname{lemma}{Lem.}{Lems.}
\crefname{theorem}{Thm.}{Thms.}
\crefname{section}{Sec.}{Secs.}
\crefname{equation}{Eq.}{Eqs.}
\crefname{table}{Tab.}{Tabs.}
\crefname{appendix}{Appendix}{Appendices}

\begin{document}

\preprint{APS/123-QED}

\title{Quantum simulation of random unitaries from Clebsch--Gordan transforms}

\author{Dmitry Grinko}
\thanks{All authors contributed equally to this work.\\
Dmitry Grinko: \href{mailto:d.grinko@uva.nl}{d.grinko@uva.nl}\\
Satoshi Yoshida: \href{mailto:satoshiyoshida.phys@gmail.com}{satoshiyoshida.phys@gmail.com}}
\affiliation{QuSoft, Amsterdam, The Netherlands}
\affiliation{Institute for Logic, Language and Computation, University of Amsterdam, The Netherlands} 
\affiliation{Korteweg-de Vries Institute for Mathematics, University of
Amsterdam, The Netherlands}

\author{Satoshi Yoshida}
\thanks{All authors contributed equally to this work.\\
Dmitry Grinko: \href{mailto:d.grinko@uva.nl}{d.grinko@uva.nl}\\
Satoshi Yoshida: \href{mailto:satoshiyoshida.phys@gmail.com}{satoshiyoshida.phys@gmail.com}}
\affiliation{Department of Physics, Graduate School of Science, The University of Tokyo, Japan}
\date{June 2026}

\begin{abstract}
    We construct exact compressed oracles for Haar-random group actions associated with an arbitrary finite-dimensional unitary representation $\rho:G\to\U(V)$ of a compact group.
    The construction is a representation-theoretic version of Zhandry's compressed-oracle technique: the memory of the oracle is stored in the Fourier basis, and each update is implemented by Clebsch--Gordan transforms.
    This framework naturally gives forward, conjugate, transpose, and inverse compressed oracles.
    For $G=\U(d)$ with the defining representation, we present efficient implementation based on high-dimensional Clebsch--Gordan transforms.
    We also explain how Ma--Huang's approximate path-recording oracle compares to our exact construction.
    For general compact groups, we describe the corresponding path-recording bases, achieved via generalized Schur transforms.
    These results clarify the relation between exact representation-theoretic compressed oracles and the path-recording bases used in algorithmic and cryptographic compressed-oracle arguments.
\end{abstract}

\maketitle

Random unitary operation offers a universal primitive for various quantum information processing including shadow tomography~\cite{huang2020predicting, zhao2021fermionic, elben2023randomized, bertoni2024shallow, wan2023matchgate, kunjummen2023shadow, levy2024classical, helsen2023shadow}, random sampling~\cite{boixo2018characterizing, arute2019quantum, ware2023sharp, hangleiter2023computational, morvan2024phase, fefferman2024anti}, randomized benchmarking~\cite{emerson2005scalable, knill2008randomized, dankert2009exact, magesan2011scalable, eisert2020quantum} and quantum random oracle model in cryptography~\cite{bellare1993random, boneh2011random, zhandry2019record, chen2024quantum, bouland2019computational, ananth2025pseudorandomness, hhan2024pseudorandom}.
It also offers a fundamental understanding of physical systems, including black holes and chaotic systems~\cite{page1993information, hayden2007black, hosur2016chaos, roberts2017chaos, nahum2017quantum, nahum2019operator, kudler2021distinguishing, brandao2021models, haferkamp2022linear, fisher2023random, suzuki2025quantum}.
Random distribution of $d$-dimensional unitary operation is modeled by the Haar measure.
There are three common ways to simulate the Haar-random unitary in the quantum circuit model: \emph{unitary $t$-design}, \emph{pseudorandom unitary (PRU)}, and \emph{compressed oracle}.

\emph{Unitary $t$-design} is given by a probability distribution $\{p_i\}_i$ on a finite set of unitary operators $\{U_i\}_i$ (we shortly write it as $\{p_i, U_i\}_i$) and it can simulate a quantum circuit having $t$ queries to a Haar-random unitary operation $U$, its complex conjugate $\bar{U}$, its transpose $U\tp$, and its inverse $U\ct$~\cite{dankert2009exact, ambainis2007quantum, mele2024introduction} [see Fig.~\ref{fig:simulation_random_unitary}~(a) and (b)].
\emph{Pseudorandom unitary (PRU)} is a probability distribution $\{p_i, U_i\}_i$ on $n$-qubit unitaries such that any polynomial-time quantum circuit cannot distinguish it from the Haar measure~\cite{ji2018pseudorandom}, where the distinguisher is allowed to query only $U_i$, or $U_i$ and $U_i\ct$, or $U_i, U_i\ct, \bar{U}_i$, and $U_i\tp$ depending on the setting.
\emph{Compressed oracle} is another way to simulate an action of random unitaries by using quantum memory to purify the classical randomness associated with the underlying measure.
It is defined as unitary operations $\fO, \cO, \tO, \iO$ called the forward, conjugate, transposed, and inverse oracles, respectively, acting on a system and auxiliary memory system.
By tracing out the memory system in the end, it simulates the forward query $U$, the conjugate query $U^*$, the transposed query $U\tp$, and the inverse query $U\ct$ of the Haar-random unitary $U$, respectively [see Fig.~\ref{fig:simulation_random_unitary}~(c)].
This is a natural generalization of the same problem for random functions~\cite{zhandry2019record}, and it was recently used in the adaptive security proof of a certain PRU construction~\cite{ma2024}, and can be used to construct cryptographic protocols such as quantum money~\cite{wiesner1983conjugate, alagic2020efficient}.
These notions are extended to several subgroups of the unitary groups, e.g., orthogonal $t$-design is defined for the orthogonal group~\cite{o2023explicit}, and pseudorandom permutation~\cite{zhandry2025note} and the compressed permutation oracle~\cite{unruh2023towards, majenz2025permutation} are defined for the permutation group.

\begin{figure}
    \centering
    \includegraphics[width=0.92\linewidth]{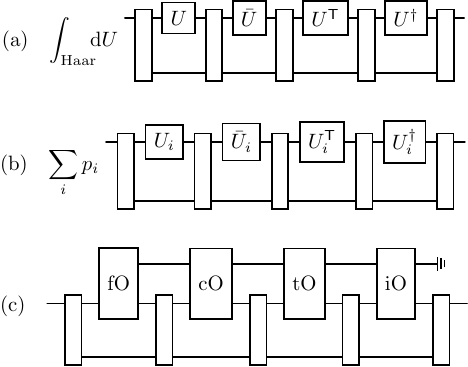}
    \caption{(a) A quantum circuit involving $t$ queries to a unitary operations $U,\bar{U},U\tp$ and $U\ct$, where $U$ is drawn from the Haar measure.
    The boxes other than $U$ represent arbitrary quantum channels.\\
    (b) Unitary $t$-design $\{p_i, U_i\}_i$ can simulate the quantum circuit (a) by using $U_i$ with probability $p_i$.\\
    (c) Our construction simulates the quantum circuit (a) exactly by using the compressed oracles $\fO,\cO,\tO,\iO$ and tracing out the auxiliary register.}
    \label{fig:simulation_random_unitary}
\end{figure}

\begin{figure*}
    \centering
    \includegraphics[width=\textwidth, page=1]{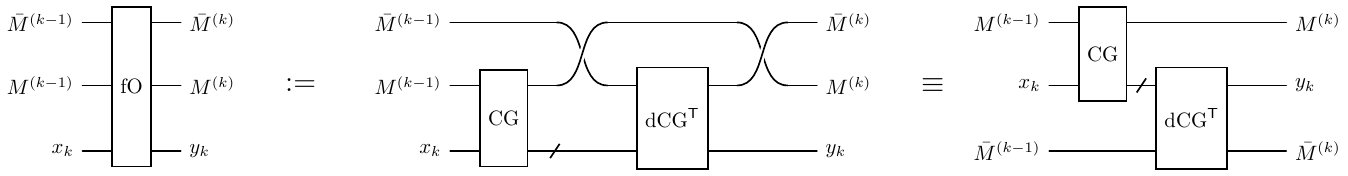}
    \caption{Compressed oracle $\fO$ for an arbitrary group $G$ written in standard quantum circuit notation (time goes from left to right).
    $x_k$ denotes the input, $y_k$---the output, while $M$ and $\bar{M}$ label basis states of the ancilla registers, which store irreducible representations of $G$.
    Clebsch--Gordan transform decomposes tensor product of representations $\lambda \otimes \rho$, where $\lambda$ is some irrep and $\rho$ is a given unitary representation of $G$.
    Such decomposition is in general not multiplicity-free, and the register carrying the multiplicity is highlighted with a slash wire.}
    \label{fig:pr_O_U}
\end{figure*}

Approximate unitary $t$-design can be implemented efficiently~\cite{emerson2005convergence, harrow2009random, brown2010convergence, brandao2016local, nakata2017unitary, nakata2017efficient, hunter2019unitary, haferkamp2021improved, ho2022exact, liu2022estimating, haferkamp2022random, jian2023linear, haferkamp2023efficient, harrow2023approximate, chen2024efficient, metger2024simple, haah2024efficient, chen2024incompressibility, belkin2024approximate, schuster2024random, laracuente2024designs, cui2025unitary} with the state-of-the-art circuit complexity $\tilde{O}(nt)$ given in Refs.~\cite{schuster2024random, laracuente2024designs, cui2025unitary} matching with the lower bound $\tilde{\Omega}(nt)$ shown in Ref.~\cite{brandao2021models}.
Known implementation of exact $t$-design is highly inefficient~\cite{nakata2021quantum} except for $t\leq 3$~\cite{dankert2009exact, webb2015clifford, zhu2017multiqubit}.
The recent breakthrough in the construction of secure PRU~\cite{ma2024}, based on the PFC construction of Ref.~\cite{metger2024simple}, shows the adaptive security proof based on the compressed oracle called the path-recording oracle simulating the Haar-random unitary approximately with the forward and inverse queries.
However, the path-recording oracle cannot choose the precision arbitrarily.
Reference~\cite{alagic2020efficient} proposes an exact simulation of the forward query of a random unitary with efficient memory size, but it is not constructive, and its efficiency in circuit complexity is not known.
The extension to subgroups of the unitary group is less clarified; e.g., an efficient construction of the compressed permutation oracle is a long-standing open problem~\cite{unruh2023towards, majenz2025permutation}.
In another line of research, the work \cite{bostanci2024generalquantumdualityrepresentations} introduced a duality between so-called Fourier subspace extraction and implementation of group representation, which was then used in the construction of quantum money.

In this work, we present an exact construction of the compressed oracle for the Haar random ensemble corresponding to any unitary representation $\rho: G\to \End (V)$ of a compact group $G$.
We provide an efficient construction for the unitary group ($G=\U(d), \, \rho(U) = U$) with the circuit complexity given by $\mathrm{poly}(n, t, \log \epsilon^{-1})$ with the compilation error $\epsilon$ and $n = \log(d)$.
Our construction is based on basic facts from the representation theory of compact groups and efficient implementation of the (dual) Clebsch-Gordan (CG) transforms for the unitary group~\cite{bacon2006efficient, bacon2007quantum, nguyen2023mixed, grinko2023gelfand, fei2023efficient, burchardt2025krovi}.
Moreover, we can simulate not only forward queries, but also conjugate, transpose and inverse queries.
This versatility is quite interesting, and in the light of recent attention to conjugate and transpose queries \cite{zhandry2025model} our work unifies simulation of all four query types.
We conjecture that this generalization could be efficiently implemented for a wide variety of groups.
In particular, for the permutation group it gives an exact Fourier-basis compressed oracle for the random-permutation model~\cite{unruh2023towards,majenz2025permutation}.
Our simulator can also be used to twirl a given quantum supermap, which can convert algorithmic errors in certain tasks to a white-noise error in higher-order quantum transformations of unitary channels~\cite{quintino2022deterministic}.

{\it Main results}.---
Now we state informally our two main results. 
First, for an arbitrary given compact group $G$ together with some unitary representation $\rho: G \rightarrow \End(V)$, where $\End(V)$ is the space of linear operators on a finite-dimensional linear space $V$.
Second, specifically for the unitary group $\U(d)$ and $\rho$ being the defining representation (labelled by Young diagram $\square$), we explain how to efficiently implement forward queries.
We briefly explain the main ideas behind the proofs, while the full proofs could be found in Appendix~\hyperref[sec:app_proof_main]{B}.

\begin{theorem}[Informal] \label{thm:main_1}
For any compact group, there exists exact compressed oracles $\fO,\mathrm{cO},\mathrm{iO}$ and $\mathrm{tO}$, which can simulate respectively forward, conjugate, inverse and transpose of an action of Haar random group elements in a given unitary representation. 
These oracles can be easily constructed from two Clebsch--Gordan transforms.
\end{theorem}
\begin{proof}[Proof idea]
We use the notation introduced in Appendix~\hyperref[sec:notation_basic_notions]{A} for the representation theory.
The main idea behind the proof technique is schematically described in \cref{fig:proof_idea} for the case of forward queries $\fO$: the quantum circuit consisting of our compressed oracles from \cref{fig:pr_O_U} can be easily seen to be equivalent to the tensor network contraction involving Clebsch--Gordan tensors, which is the tensor network representation of the Clebsch--Gordan transform.
Namely, top and bottom tensor networks in the middle of the equation in Fig.~\ref{fig:proof_idea} correspond to matrix units $E^\lambda_{T,S}$ of the commutant of the tensor action $\rho^{\otimes t}$.
The summation is done over all possible matrix units $E^\lambda_{T,S}$ of the commutant for $\lambda\in \hat{G}^{(t)}$, $S = (S_0, s_1, S_1, \ldots, s_t, S_t), T = (T_0, t_1, T_1, \ldots, t_t, T_t) \in B(\lambda)$.
To finish the proof, \cref{thm:app_forward_oracle} identifies the oracle contraction with the sum of matrix units, and \cref{lemma:mat_units} identifies that sum with the required Haar integral.
The formal extension to arbitrary mixed sequences of forward, conjugate, transpose, and inverse queries is given in \cref{thm:app_all_query_oracles}.
The high-dimensional data needed for mixed $\U(d)$ Clebsch--Gordan transforms is discussed in \hyperref[app:mixed_cg_transforms]{Appendix~D.1}.
The full proof can be found in Appendix~\hyperref[sec:app_proof_main]{B}.
\end{proof}

The same construction also has a useful database-basis interpretation.
The Fourier memory of the oracle is naturally organized by the Peter--Weyl theorem: after $t$ forward queries it consists of a pair of matched Peter--Weyl registers for the degree-$t$ matrix coefficients of $\rho^{\otimes t}$.
In the appendices we describe how, for $G=\U(d)$, this Fourier basis can be changed to a basis of unordered input-output records using Schur transforms.
For general compact groups this motivates an abstract database space obtained from the commutant-invariant part of $V^{\otimes t}\otimes \bar V^{\otimes t}$, which is the representation-theoretic analogue of Ma-Huang path-recording for arbitrary groups.

\begin{figure*}[htbp]
    \centering
    \includegraphics[width=0.8\textwidth, page=2]{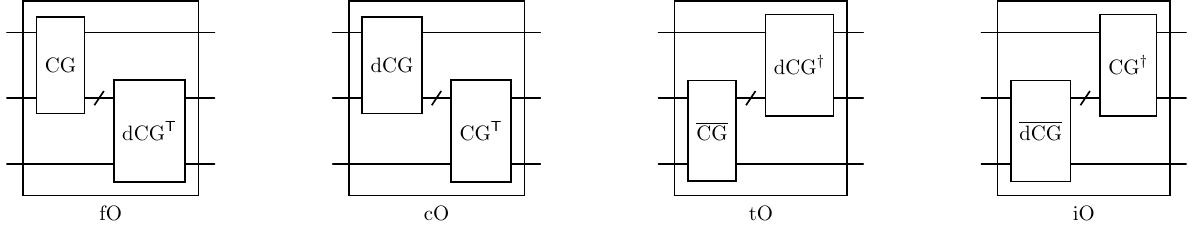}
    \caption{Four types of oracles based on the corresponding query types: forward $\fO$, conjugate $\cO$, transpose $\tO$, inverse $\iO$.}
    \label{fig:4_types}
\end{figure*}

\begin{figure*}[htbp]
    \centering
    \includegraphics[width=\textwidth, page=3]{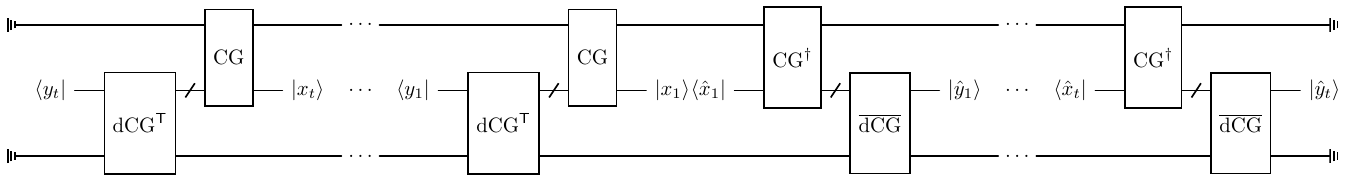}
    \vspace{2pt} $=$
    \large $$\hspace*{5pt}\sum_{\lambda \in \widehat{G}^{(t)}} \sum_{\substack{T,S \in B(\lambda)}}\hspace*{-10pt} \includegraphics[width=0.9\textwidth, valign=c, page=11]{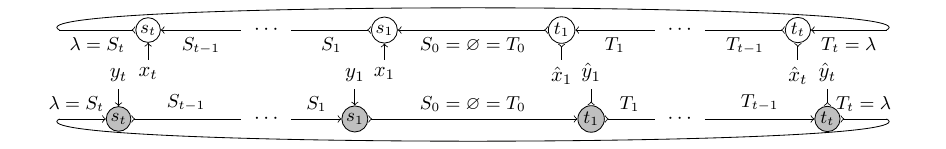}$$
    \normalsize $=$\vspace{-5pt}
    \large $$\int_{\text{Haar}} dg \, \bra{y_t}\rho(g)\ket{x_t} \dotsb \bra{y_1}\rho(g)\ket{x_1} \cdot \bra{\hat{x}_1}\rho(g^{-1})\ket{\hat{y}_1} \dotsb \bra{\hat{x}_t}\rho(g^{-1})\ket{\hat{y}_t}$$
    \caption{
    Proof idea behind the construction of our compressed oracle $\fO$. 
    The top figure is drawn in the Heisenberg picture and represents $t$ queries to the oracle $\fO$.
    The middle figure is the rewriting of the top one in terms of (dual) Clebsch--Gordan tensors, which comprise matrix units of the commutant of $\rho^{\otimes t}$ action.
    The bottom is an equivalent Haar integral expression.
    The equalities are proven in detail in Appendix~\hyperref[sec:app_proof_main]{B} of the SM~\cite{supple}.
    }
    \label{fig:proof_idea}
\end{figure*}

\begin{theorem} \label{thm:main_2}
    Successive application of forward oracles $\fO$ can simulate $t$ queries of the Haar random unitary group $\U(d)$ elements with total gate and depth complexity $t^5 \mathrm{polylog}(d, \epsilon^{-1})$. The memory cost is $t^2 \mathrm{polylog}(d, \epsilon^{-1})$.
\end{theorem}
\begin{proof}[Proof idea]
The main idea behind this theorem is that for the case of the unitary group $\U(d)$, the construction of the $t$-th CG transform is efficient in $t$ and $d$, having $\widetilde{O}(t^4)$ gate and depth complexity and $\widetilde{O}(t^2)$ memory complexity, where $\widetilde{O}(\cdot)$ hides polylogarithmic factors in $d$ and $\epsilon^{-1}$.
See Appendix~\hyperref[app:cg_trans]{D} of the Supplementary Material (SM)~\cite{supple} for details.
If we have in total $t$ calls to simulate, then the total time and depth complexity is given by $\sum_{k=1}^t \widetilde{O}(k^4) = \widetilde{O}(t^5)$, while the memory complexity is $\widetilde{O}(t^2)$ since we are reusing and adding new memory ``on the fly''.
This construction relies on the fact that we can efficiently and reversibly compress Gelfand--Tsetlin patterns, which label basis vectors of irreducible representations (irreps) of the unitary group $\U(d)$.
\end{proof}

\begin{figure*}[htbp]
    \centering
    \includegraphics[width=\linewidth]{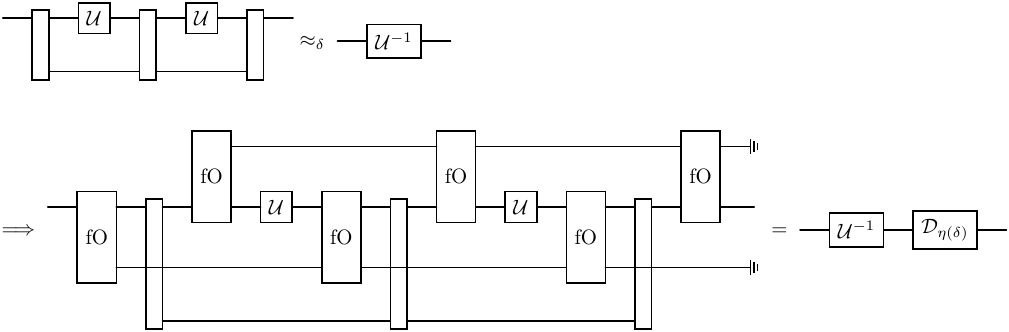}\\
    \caption{Twirling of the approximate unitary inversion protocol using the forward compressed oracle $\fO$.
    The top figure corresponds to a quantum comb that approximately implements unitary inversion with $n$ queries to the input unitary channel $U$ with the average-case channel fidelity $F = 1-\delta$ defined in Eq.~\eqref{eq:fidelity}.
    The bottom figure corresponds to the twirled quantum comb that transforms $n$ queries of $U$ into the channel given by $\mathcal{D}_{\eta(\delta)} \circ \mathcal{U}^{-1}$, where $\mathcal{D}_{\eta(\delta)}$ is a depolarizing channel with the noise parameter given by $\eta(\delta) = {d^2 \over d^2-1} \delta$.}
    \label{fig:twirling_inversion}
\end{figure*}

{\it Comparison with the previous works}.---
We compare the construction of the compressed oracles for $t$ queries of the Haar-random $n$-qubit unitary with the previous works~\cite{alagic2020efficient, ma2024} (see Tab.~\ref{tab:comparison}).
The exact simulation of the Haar-random unitary shown in Ref.~\cite{alagic2020efficient} uses $O(nt)$ qubits, but it is not constructive, and the depth is not bounded.
The path-recording oracle shown in Ref.~\cite{ma2024} is efficient in memory and depth, but its precision is fixed to be $O(t^2/2^n)$, where the precision is given by the worst-case diamond-norm error between the quantum channels given in Figs.~\ref{fig:simulation_random_unitary}~(a) and (c), where we take the worst case with respect to the quantum channels inserted in between the random unitaries.
The database-basis form in Appendix~\hyperref[sec:exact_path_recording_oracle]{C} makes this relation explicit.
After changing the basis of our Fourier ancilla, the exact oracle has a path-recording support rule, and the Ma--Huang oracle is an approximation to this exact oracle.
Our construction provides efficiency in memory and depth, and its precision can be chosen arbitrarily since the error only comes from the compilation error of implementing the CG transforms.

{\it Application for the quantum cryptography}.---
Our random unitary simulator can be applied\footnote{Oracle accesses to $V_g$ and $V_{g^{-1}}$ are equivalent to oracle accesses to $U_g$ and $U_{g^{-1}}$ for the following more standard oracle, after identifying the alphabet with $\mathbb Z_d$:
    \begin{align}
    \label{eq:def_Ug}
        U_g\coloneqq \sum_{x,y\in\mathbb Z_d} \ketbra{x, y\oplus_d g(x)}{x,y},
    \end{align}
    where $\oplus_d$ represents addition in $\mathbb Z_d$, since either can be simulated using two queries to the other [see \hyperref[sec:equivalence_permutation_oracle]{Appendix~E.2} of the SM~\cite{supple} for the details].} for the permutation group $G = \mathfrak{S}_{d}$ with the representation
\begin{align}
\label{eq:def_Vg}
    g\in \mathfrak{S}_{d}\mapsto V_g\coloneqq \sum_{x\in\mathbb Z_d} \ketbra{g(x)}{x} \in \End(\mathbb{C}^d).
\end{align}
This construction gives an exact Fourier-basis simulator for a Haar-random permutation.
After the ordinary group QFT on $\mathfrak{S}_d$, the memory has two Peter--Weyl registers $\ket{\lambda,P,Q}$.
Each query is updated by applying the Clebsch--Gordan transform for tensoring with the defining permutation representation to both Peter--Weyl indices.
An efficient implementation would therefore follow from efficient reversible implementations of these $\mathfrak{S}_d$ Clebsch--Gordan transforms.

{\it Application for the quantum supermaps}.---
Reference~\cite{quintino2022deterministic} shows that any approximate unitary inversion protocol can be converted to a unitary inversion protocol with the white-noise error.
This conversion is done by twirling the corresponding Choi matrix, but it is unknown how the conversion is done on the level of a quantum circuit.
Our simulator makes this conversion possible as shown in Fig.~\ref{fig:twirling_inversion}, which uses the forward compressed oracle $\fO$ to implement the twirling of the quantum comb.
Suppose we have a quantum comb $\mathcal{C}$ approximately implementing unitary inversion with $n$ queries to the input unitary channel $U\in U(d)$ with the average-case channel fidelity $F = 1-\delta$, where $F$ is defined by
\begin{align}
    \label{eq:fidelity}
    F \coloneqq \int \dd U F_{\mathrm{ch}}(\mathcal{C}[\mathcal{U}^{\otimes n}], \mathcal{U}^{-1}),
\end{align}
where $\dd U$ is the Haar measure of $\U(d)$, $\mathcal{U}(\cdot) = U \cdot U^{-1}$ and $\mathcal{U}^{-1}(\cdot) = U^{-1} \cdot U$ are the unitary channel corresponding to $U$ and $U^{-1}$, and $F_{\mathrm{ch}}(\mathcal{C}[\mathcal{U}^{\otimes n}], \mathcal{U}^{-1})$ is the channel fidelity given by $F_{\mathrm{ch}}(\mathcal{C}[\mathcal{U}^{\otimes n}], \mathcal{U}^{-1})\coloneqq {1\over d^2} \sum_i \abs{\Tr(K_i U)}^2$ using the Kraus operators $\{K_i\}$ satisfying $\mathcal{C}[\mathcal{U}^{\otimes n}](\cdot) = \sum_i K_i \cdot K_i^\dagger$.
Then, we can construct a quantum comb that transforms $n$ queries of $U$ into the channel given by
\begin{align}
    \mathcal{D}_{\eta(\delta)} \circ \mathcal{U}^{-1},
\end{align}
where $\mathcal{D}_{\eta(\delta)}$ is a depolarizing channel defined by
\begin{align}
    \mathcal{D}_{\eta(\delta)}(\cdot) = [1-\eta(\delta)] \cdot + \eta(\delta) {\1_d \over d} \Tr[\cdot]
\end{align}
with the noise parameter given by $\eta(\delta) = {d^2 \over d^2-1} \delta$.
This conversion can be extended to other tasks, such as unitary transposition~\cite{quintino2022deterministic} and unitary complex conjugation~\cite{miyazaki2019complex,ebler2023optimal} using the corresponding compressed oracles.
The conversion to the white-noise error is beneficial in two ways.
First, this conversion makes the algorithmic error independent of the input unitary channel $U$, and the worst-case error of the converted protocol is the same as the average-case error of the original protocol~\cite{quintino2022deterministic}.
Secondly, the white-noise error can be treated more easily than the general error, e.g., the white-noise error can be mitigated in a cost-optimal way using the rescaling method~\cite{tsubouchi2023universal}, and the white-noise error on the quantum state can be efficiently purified by using multiple rounds of the swap test~\cite{childs2025streaming} or using the Schur sampling~\cite{li2024optimal}.

{\it Conclusion.}---
This Letter presents an exact implementation of the Haar-random ensemble of any unitary representation of a compact group using the (dual) CG transforms.
This construction provides an efficient simulation of the Haar-random unitary with arbitrary precision.
For the permutation group, our construction gives an exact Fourier-basis compressed oracle, whose efficient implementation reduces to efficient $\mathfrak{S}_d$ Clebsch--Gordan transforms.
Our construction can also be used to implement the twirling of the quantum supermap at the circuit level, which can be used to transform the algorithmic error into the white-noise error.

Our work shows that an efficient simulation of the Haar-random ensemble is possible once the corresponding (dual) CG transforms are implemented efficiently.
Moreover, our results highlight the quantum information theoretic importance of the CG transforms and motivates studying the CG transforms beyond the unitary group.
We leave it as an open problem to provide efficient implementations of the CG transforms for several groups, such as the permutation group.

\begin{table}
    \centering
    \caption{Comparison of spacetime cost for simulation of $t$ queries of $n$-qubit Haar random unitary with Ref.~\cite{alagic2020efficient} and the path-recording oracle in Ref.~\cite{ma2024}.
    Since the results of Ref.~\cite{alagic2020efficient} are not constructive, the depth and precision are not studied. Error $\epsilon$ is an artifact of the compilation process: in particular, our construction can be made exact in theory (i.e. $\epsilon = 0$), while Ref.~\cite{ma2024} always has an inherent error of order $O(t^2/2^n)$.}
    \begin{ruledtabular}
    \begin{tabular}{c|c|c|c}
         & Memory & Depth & Precision \\\hline
        Ref.~\cite{alagic2020efficient} & $O(nt)$ & --- & --- \\
        Ref.~\cite{ma2024} & $n t \cdot \mathrm{polylog} (\epsilon^{-1})$ & $\mathrm{poly}(t,n, \log \epsilon^{-1})$ & $\epsilon + O(t^2/2^n)$\\
        This work & $t^2 \mathrm{poly}(n, \log \epsilon^{-1})$ & $t^5 \mathrm{poly}(n, \log \epsilon^{-1})$ & $\epsilon$
    \end{tabular}
    \end{ruledtabular}
    \label{tab:comparison}
\end{table}

{\it Note added.}---
During the preparation of the first version of this manuscript, Carolan~\cite{carolan2025compressed} independently constructed a different compressed permutation oracle.
We compare Carolan's construction with ours in \cref{sec:permutation_fourier_database}.
After the completion of the first version, Barak Nehoran pointed out to us that the idea behind our compressed oracle had already independently appeared, in an abstract form, in Sec.~8.1.3 of Ref.~\cite{harrow2005applications}.
We thank him for bringing that to our attention.
While preparing the second version of this manuscript, we also became aware of the concurrent and independent work~\cite{FLMNW26}, whose results overlap with the new results added in the second version.
The work~\cite{FLMNW26} further applies the framework to prove security of a new pseudorandom unitary construction.

{\it Acknowledgments.}---
We thank Mio Murao, Shogo Yamada, Adam Burchardt, Maris Ozols, Yu-Hsuan Huang, Gina Muuss, Silvia Ritsch, Christian Schaffner, Michael Walter, Christian Majenz and Barak
Nehoran for fruitful discussions.
We thank the authors of Ref.~\cite{FLMNW26} for the coordination.
We acknowledge AI assistance (ChatGPT 5.5 Pro) in revising the second version of the manuscript.
All key ideas were developed by the authors without AI assistance.
S.~Y.\ acknowledges support by Japan Society for the Promotion of Science (JSPS) KAKENHI Grant Number 23KJ0734, FoPM, WINGS Program, the University of Tokyo, and DAIKIN Fellowship Program, the University of Tokyo.
D.~G. acknowledges support by NWO grant NGF.1623.23.025 (“Qudits in theory and experiment”) and NWO Vidi grant (Project No. VI.Vidi.192.109).

\let\oldaddcontentsline\addcontentsline
\renewcommand{\addcontentsline}[3]{}
\bibliography{references}
\let\addcontentsline\oldaddcontentsline

\clearpage
\onecolumngrid
\appendix
\makeatletter
\def\@hangfrom@appendix#1#2#3{%
 #1%
 \@if@empty{#2}{%
  #3%
 }{%
  #2.\ \@if@empty{#3}{}{#3}%
 }%
}%
\def\@appendixcntformat#1{\csname the#1\endcsname}%
\let\@hangfrom@section\@hangfrom@appendix
\let\@sectioncntformat\@appendixcntformat
\makeatother
\setcounter{secnumdepth}{2}
\setcounter{tocdepth}{2}
\renewcommand{\thesubsection}{\thesection.\arabic{subsection}}
\makeatletter
\renewcommand{\p@subsection}{}
\makeatother
\renewcommand{\theHsubsection}{\theHsection.\arabic{subsection}}
\renewcommand{\theHequation}{\theHsection.\arabic{equation}}

\tableofcontents

\section{Notation and basic notions}
\label{sec:notation_basic_notions}
In this section, we summarize the basics and notations of the representation theory.
We mostly focus on two groups: unitary group $\U(d)$ and symmetric group $\mathfrak{S}_d$.
Irreps of these groups are labelled by highest weights and partitions.

A partition $\lambda=(\lambda_1,\lambda_2,\ldots)$ of $t$, denoted $\lambda\vdash t$, is a weakly decreasing sequence of nonnegative integers with $|\lambda|\coloneqq\sum_i\lambda_i=t$.
Its Young diagram is the left-justified array with $\lambda_i$ boxes in row $i$, and its length $\ell(\lambda)$ is the number of nonzero rows.
For polynomial irreps of $\U(d)$, the label $\lambda$ is a Young diagram with $\ell(\lambda)\leq d$.
For irreps of $\mathfrak{S}_t$, the label $\lambda$ is a Young diagram with $|\lambda|=t$.
A standard Young tableau of shape $\lambda$ is a filling of the boxes of $\lambda$ by $1,\ldots,|\lambda|$ that strictly increases along each row and each column.
We denote the set of standard Young tableaux of shape $\lambda$ by $\SYT(\lambda)$.
Equivalently, a tableau $T\in\SYT(\lambda)$ is a path in the Bratteli diagram for the tower $\mathfrak{S}_0\subset\mathfrak{S}_1\subset\dotsb\subset\mathfrak{S}_{|\lambda|}$, where the path adds the box containing $i$ at step $i$.
These tableaux label the usual Young orthonormal basis of the Specht module of $\mathfrak{S}_{|\lambda|}$.

A semistandard Young tableau of shape $\lambda$ and alphabet $[d]$ is a filling by symbols in $[d]$ that weakly increases along rows and strictly increases down columns.
We denote the set of such tableaux by $\SSYT(\lambda)$.
Equivalently, this basis can be labelled by Gelfand--Tsetlin patterns $\GT(\lambda)$.
A Gelfand--Tsetlin pattern of shape $\lambda$ is an interlacing sequence of partitions
\begin{align}
    L
    =
    \bigl(\varnothing=L_0\sqsubseteq L_1\sqsubseteq\dotsb\sqsubseteq L_d=\lambda\bigr)
    \in\GT(\lambda),
\end{align}
where the interlacing condition is defined as 
\begin{align}
    L_{k-1} \sqsubseteq L_k
    \quad\Longleftrightarrow\quad
    L_{k,1}\geq L_{k-1,1}\geq L_{k,2}\geq L_{k-1,2}\geq\dotsb\geq L_{k,k}
    \quad\forall k\in[d].
\end{align}
$\GT(\lambda)$ denotes the set of all such interlacing patterns $L$.
The bijection between $\SSYT(\lambda)$ and $\GT(\lambda)$ sends a tableau to the chain of shapes obtained by keeping only entries at most $k$ for each $k\in[d]$.
In the later formulas, we usually write $L$ for a Gelfand--Tsetlin pattern and $T$ for a Bratteli path, equivalently a standard Young tableau.

The tensor product $\rho^{\otimes t}(g)$ of the given representation $\rho: G\to \End(V)$ is decomposed as
\begin{align}
    V^{\otimes t} &\simeq \bigoplus_{\lambda\in \hat{G}^{(t)}} V_\lambda \otimes M_\lambda,\\
    \USch \rho(g)^{\otimes t}\USch^\dagger &= \bigoplus_{\lambda\in \hat{G}^{(t)}} \rho_\lambda(g) \otimes I_{M_{\lambda}} \quad \forall g\in G, \label{def:schur_decomp}
\end{align}
where $V_\lambda$ is the representation space of an irrep $\rho_\lambda$, $M_\lambda$ is the corresponding multiplicity space, and $\hat{G}^{(t)}$ is the set of irrep labels appearing in the decomposition.
The unitary $\USch:V^{\otimes t}\to \bigoplus_{\lambda\in\hat{G}^{(t)}}V_\lambda\otimes M_\lambda$ denotes the Schur transform for the group $G$ and its represenation $\rho$.

$M_\lambda$ can be also viewed as irreps of the commutant algebra $\End_G(V^{\otimes t})$ of the $\rho^{\otimes t}$ action, where the commutant is formally defined as 
\begin{equation}
    \End_G(V^{\otimes t}) \defeq \{X \in \End(V^{\otimes t}) \mid [X, \rho(g)^{\otimes t}] = 0, \forall g \in G\} 
\end{equation}

The Schur transform in \cref{def:schur_decomp} is obtained by recursively applying the Clebsch--Gordan transform $\CG$, which corresponds to the irreducible decomposition of $V_\mu \otimes V$ for $\mu \in \hat{G}^{(t-1)}$ by
\begin{align}
    \CG (\rho_\mu(g) \otimes \rho(g)) \CG^\dagger = \bigoplus_{\nu\in \hat{G}^{(t)}} \rho_\nu(g) \otimes I_{m_{\nu, \mu}} \quad \forall g\in G,
\end{align}
Here $m_{\nu, \mu}$ is the multiplicity of the irrep $V_\nu$ in $V_\mu \otimes V$.
In tensor-network diagram notation, we draw the corresponding Clebsch--Gordan tensor and its adjoint as
\begin{align}
    \label{eq:cg_tensor_network}
    \bra{\nu,N} \CG \ket{\mu,M,\rho,R} &= \includegraphics[width=0.16\textwidth, valign=c, page=3]{figures/tensor_networks.pdf}
    \\
    \bra{\mu,M,\rho,R} \CG^\dagger \ket{\nu,N}  &= \overline{\bra{\nu,N} \CG \ket{\mu,M,\rho,R}} = \includegraphics[width=0.16\textwidth, valign=c, page=4]{figures/tensor_networks.pdf}.
\end{align}
Sometimes, if the irrep labels are clear from the context then we write $\bra{N} \CG \ket{M,R}$ for brevity.

Similarly, we define the dual Clebsch--Gordan tensor and its adjoint by
\begin{align}
    \label{eq:dcg_tensor_network}
    \bra{\mu,M}\dCG\ket{\nu,N,\bar{\rho},\bar{R}} &= \includegraphics[width=0.16\textwidth, valign=c, page=5]{figures/tensor_networks.pdf}
    \\
    \bra{\nu,N,\bar{\rho},\bar{R}}\dCG^\dagger\ket{\mu,M} &= \overline{\bra{\mu,M}\dCG\ket{\nu,N,\bar{\rho},\bar{R}}} = \includegraphics[width=0.16\textwidth, valign=c, page=6]{figures/tensor_networks.pdf}.
\end{align}
By tensoring representation $\rho$ several times, we obtain the multiplicity space $M_\lambda$ whose basis vector is labeled by
\begin{align}
    T \in B(\lambda) \coloneqq\{(T_0, t_1, T_1, \ldots, t_t, T_t) \mid T_i\in \hat{G}^{(i)}, T_t=\lambda, t_i\in [m_{T_i, T_{i-1}}]\text{ for }i\in[t]\}.
\end{align}
$B(\lambda)$ is the set of paths in the Bratteli diagram for the commutant of the $\rho^{\otimes t}$ from the trivial representation to $\lambda$, where the edge label $t_i$ records the multiplicity of the transition $T_{i-1}\to T_i$.

For $\lambda\in\hat{G}^{(t)}$ and $T,S\in B(\lambda)$, the matrix unit $E^\lambda_{T,S}$ is defined by its nonzero Schur block
\begin{align}
    \label{def:matrix_units_blcok_schur}
    \USch E^{\lambda}_{T,S}\USch^\dagger
    =
    I_{V_\lambda}\otimes\ketbra{T}{S}_{M_\lambda}.
\end{align}
By Schur's lemma, the commutant of $\rho^{\otimes t}$ is the linear span of these matrix units:
\begin{align}
    \End_G(V^{\otimes t}) = \mathrm{span}\{E^{\lambda}_{T,S} \mid \lambda\in \hat{G}^{(t)}, T,S\in B(\lambda)\}.
\end{align}
We will use the following tensor-network representation of these matrix units.
For paths $S=(S_0,s_1,S_1,s_2,S_2,\dotsc)$ and $T=(T_0,t_1,T_1,t_2,T_2,\dotsc)$ in $B(\lambda)$, the matrix unit $E^\lambda_{S,T}$ is
\begin{equation}
    \label{eq:matrix_unit_tensor_network}
    E^\lambda_{S,T} \defeq \includegraphics[width=0.23\textwidth, valign=c, page=7]{figures/tensor_networks.pdf},  \qquad  \qquad  \bra{x} E^\lambda_{S,T}\ket{\hat{x}} = \includegraphics[width=0.3\textwidth, valign=c, page=8]{figures/tensor_networks.pdf}.
\end{equation}
Here $s_i,t_i$ are multiplicity labels, while $S_i,T_i$ are irrep labels.

We will also use the following identity relating normal and dual Clebsch--Gordan tensors \cite[Eq.\ (10), p.\ 289]{Vilenkin1992}:
\begin{equation}
    \label{eq:dual_normal_CG}
    \includegraphics[width=0.2\textwidth, valign=c, page=3]{figures/tensor_networks.pdf} \; = \; \sqrt{\frac{d_\mu}{d_\nu}} \quad \includegraphics[width=0.2\textwidth, valign=c, page=6]{figures/tensor_networks.pdf},
\end{equation}
where $i$ denotes the multiplicity label.

Finally, we use the notion of vectorization of linear operator $X: \mathcal{I} \to \mathcal{O}$, which is defined by
\begin{align}
    \dket{X} \coloneqq \sum_{i} \ket{i}_{\mathcal{I}} \otimes (X\ket{i})_{\mathcal{O}}
\end{align}
using the computational basis $\{\ket{i}\}_i$ of $\mathcal{I}$.

\section{\texorpdfstring{Proof of Theorem~\protect\ref{thm:main_1}}{Proof of Theorem 1}}
\label{sec:app_proof_main}

In this section and the rest of the Appendix, we assume that the reader is familiar with standard representation theoretic notions. 
To state our main result, we need the following lemma, which is inspired by \cite[Theorem 1]{cioppa2013matrix}: 

\begin{lemma}\label{lemma:mat_units}
    For a given compact group $G$ and a unitary representation $\rho: G \rightarrow \End(V)$, we have the following relation:
    \begin{align}
        \label{eq:mat_units}
        &\int_{\mathrm{Haar}} dg \rho(g)_{y_n,x_n} \dots \rho(g)_{y_1,x_1} \rho(g^{-1})_{\hat{x}_1,\hat{y}_1} \dotsc \rho(g^{-1})_{\hat{x}_n,\hat{y}_n} = \sum_{\lambda \in \hat{G}^{(n)}} \frac{1}{d_\lambda} \sum_{T,S \in B(\lambda)} \bra{\hat{x}} E^{\lambda}_{T,S} \ket{x} \bra{y} E^{\lambda}_{S,T} \ket{\hat{y}},
    \end{align}
    where $B(\lambda)$ is a set of labels of some orthonormal basis in $M_\lambda$, and $E^{\lambda}_{T,S}$ is a set of orthogonal matrix units for the commutant of $\rho^{\otimes n}$ action.
\end{lemma}
\begin{proof}
    Firstly, we rewrite the left-hand side of Eq.~\eqref{eq:mat_units} by using the following identity:
    \begin{align}\label{eq:start_point}
        \rho(g)_{y_n,x_n} \dots \rho(g)_{y_1,x_1} \rho(g^{-1})_{\hat{x}_1,\hat{y}_1} \dotsc \rho(g^{-1})_{\hat{x}_n,\hat{y}_n}
        =
        \Tr \sof*{\ket{x}\bra{y} \otimes \ket{\hat{x}}\bra{\hat{y}} \cdot \rho(g)^{\otimes n} \otimes {\bar{\rho}(g)}^{\otimes n}},
    \end{align}
    where we used unitarity of the representation $\rho$: $\rho(g^{-1})\tp = \bar{\rho}(g)$, where $\bar{\rho}(g)$ denotes complex conjugate of $\rho(g)$.
    According to the Peter--Weyl theorem
    \begin{align}
        V^{\otimes n} \simeq^{\USch} \bigoplus_{\lambda \in \hat{G}^{(n)}} V_\lambda \otimes M_\lambda,
    \end{align}
    where $\hat{G}^{(n)}$ is the set of irreps in the tensor product representation $V^{\otimes n}$, $V_\lambda$ is the irrep of $G$ with the label $\lambda$, and $M_\lambda$ is the multiplicity space of $\rho^{\otimes n}$ action.
    The basis transformation is achieved via unitary matrix $\USch$ called the Schur transform.
    Similar Peter--Weyl theorem holds for the dual representation $\bar{V}$ of $\rho$:
    \begin{align}
        \bar{V}^{\otimes n} \simeq^{\bUSch} \bigoplus_{\lambda\in \hat{G}^{(n)}} V_{\bar{\lambda}} \otimes M_{\lambda},
    \end{align}
    where $\bar{\lambda}$ is the label of the irrep corresponding to the dual representation of $\lambda$.
    Now we apply two Schur transforms $\USch$ to block-diagonalise the whole operator according to the Peter--Weyl decomposition of both left and right parts of our space $V^{\otimes n} \otimes \bar{V}^{\otimes n}$:
    \begin{align}
        \int_{\text{Haar}} dg \of{\USch \rho(g)^{\otimes n} \USch^\dagger} \otimes \of{\bUSch {\bar{\rho}(g)}^{\otimes n} \bUSch^\dagger}  &= \int_{\text{Haar}} dg \of*{\bigoplus_{\lambda} \rho_\lambda(g) \otimes I_\lambda} \otimes \of*{\bigoplus_{\lambda'} \bar{\rho}_{\lambda'}(g) \otimes I_{\lambda'}}
    \end{align}
    By using the grand orthogonality relations
    \begin{equation}
        \int_{\text{Haar}} dg \rho_\lambda(g)_{x,y} \bar{\rho}_{\lambda'}(g)_{x',y'} = \frac{1}{d_\lambda} \delta_{\lambda,\lambda'} \delta_{x,x'} \delta_{y,y'},
    \end{equation}
    we get
    \begin{align}
        \int_{\text{Haar}} dg \of*{\bigoplus_{\lambda} \rho_\lambda(g) \otimes I_\lambda} \otimes \of*{\bigoplus_{\lambda'} \bar{\rho}_{\lambda'}(g) \otimes I_{\lambda'}}
        &= \bigoplus_{\lambda \in \hat{G}^{(n)}}
        \frac{1}{d_\lambda} \proj{\Phi^+_\lambda} \otimes I_\lambda \otimes I_{\lambda} \\
        &= \bigoplus_{\lambda \in \hat{G}^{(n)}}
        \frac{1}{d_\lambda} \sum_{T,S \in B(M_\lambda)} \proj{\Phi^+_\lambda} \otimes \ket{S,\bar{T}}\bra{S,\bar{T}}
    \end{align}
    where $\ket{\Phi^+_\lambda} \defeq \dket{\rho_\lambda(e)}$ is the vectorisation of identity operator on the irrep $\lambda$. 
    Using the definition of matrix units from \cref{def:matrix_units_blcok_schur}, after vectorizing the second factor in the conjugate Schur basis, we get
    \begin{equation}
        \dketbra{\USch E^{\lambda}_{S,T} \USch^\dagger} = \proj{\Phi^+_\lambda} \otimes \ket{S,\bar{T}}\bra{S,\bar{T}}.
    \end{equation}
    Therefore, by combining everything we get:
    \begin{align}
        \int_{\text{Haar}} &dg \, \Tr \sof*{ \of{\ket{x}\bra{y} \otimes \ket{\hat{x}}\bra{\hat{y}}} \cdot \of{\rho(g)^{\otimes n} \otimes {\bar{\rho}(g)}^{\otimes n}}} = \\
        &= \Tr \sof*{ \of{ \USch \ket{x}\bra{y} \USch^\dagger \otimes \bUSch \ket{\hat{x}}\bra{\hat{y}} \bUSch^\dagger} \int_{\text{Haar}} dg \of{\USch \rho(g)^{\otimes n} \USch^\dagger} \otimes \of{\bUSch {\bar{\rho}(g)}^{\otimes n} \bUSch^\dagger}} \\
        &= \sum_{\lambda \in \hat{G}^{(n)}} \frac{1}{d_\lambda} \sum_{T,S \in B(\lambda)} \Tr \sof*{ \of{\USch \ket{x}\bra{y} \USch^\dagger \otimes \bUSch \ket{\hat{x}}\bra{\hat{y}} \bUSch^\dagger} \dketbra{\USch E^{\lambda}_{S,T} \USch^\dagger }} \\
        &= \sum_{\lambda \in \hat{G}^{(n)}} \frac{1}{d_\lambda} \sum_{T,S \in B(\lambda)} \Tr \sof*{ \of{\USch \ket{x}\bra{y} \USch^\dagger \otimes \bUSch \ket{\hat{x}}\bra{\hat{y}} \bUSch^\dagger } \of{\USch \otimes \bUSch }\dketbra{ E^{\lambda}_{S,T}} \of{\USch^\dagger \otimes \bUSch^\dagger }  } \\
        &= \sum_{\lambda \in \hat{G}^{(n)}} \frac{1}{d_\lambda} \sum_{T,S \in B(\lambda)} \Tr \sof*{ \of{ \ket{x}\bra{y} \otimes \ket{\hat{x}}\bra{\hat{y}}} \dketbra{ E^{\lambda}_{S,T}} } \\
        &= \sum_{\lambda \in \hat{G}^{(n)}} \frac{1}{d_\lambda} \sum_{T,S \in B(\lambda)} \langle{y,\hat{y}} \dketbra{ E^{\lambda}_{S,T}} {x,\hat{x}} \rangle   \\
        &= \sum_{\lambda \in \hat{G}^{(n)}} \frac{1}{d_\lambda} \sum_{T,S \in B(\lambda)} \bra{\hat{x}} E^{\lambda}_{T,S} \ket{x} \bra{y} E^{\lambda}_{S,T} \ket{\hat{y}},
    \end{align}
    which completes the proof.
\end{proof}
This lemma was proven for the unitary group in \cite{cioppa2013matrix}. 
Similar statement trivially holds also for finite groups and their unitary representations.
Now we are ready to present our main theorem, which is a formal version of \Cref{thm:main_1}:
\begin{theorem}\label{thm:app_forward_oracle}
Consider the compressed oracle $\fO$ for a given arbitrary compact group $G$ together with its representation $\rho$, defined in \cref{fig:pr_O_U}.
Then the following equality is true:
    \begin{align}
        \label{eq:thm_main_2}
        &\sum_{\lambda \in \hat{G}^{(n)}} \frac{1}{d_\lambda} \sum_{T,S \in B(\lambda)} \bra{\hat{x}} E^{\lambda}_{T,S} \ket{x} \bra{y} E^{\lambda}_{S,T} \ket{\hat{y}} = \Tr_{\bullet} \sof*{\fO_{\bullet,(y_n,x_n)} \dotsb \fO_{\bullet,(y_1,x_1)} \proj{\varnothing}_{\bullet} \fO^\dagger_{\bullet,(\hat{x}_1,\hat{y}_1)} \dotsb \fO^\dagger_{\bullet,(\hat{x}_n,\hat{y}_n)} }
    \end{align}
    where $B(\lambda)$ is a basis of irrep $\lambda$ of the commutant of $\rho^{\otimes n}$ action, $E^{\lambda}_{T,S}$ are matrix units of the commutant, and $\fO_{\bullet,(y_k,x_k)} := (I \otimes \bra{y_k})\fO(I \otimes \ket{x_k})$, and compressed oracle $\fO$ acts on auxiliary and working registers.
\end{theorem}
\begin{proof}
    As first step, we insert resolutions of identities on the multiplicity registers, and we redraw the RHS of \cref{eq:thm_main_2} as in \cref{fig:proof_idea}, where white circles correspond to Clebsch--Gordan tensor and grey circles correspond to dual Clebsch--Gordan tensors \cite{grinko2023gelfand}.
    By the tensor-network representation of matrix units in Eq.~\eqref{eq:matrix_unit_tensor_network}, the $x,\hat{x}$ tensor network in \cref{fig:proof_idea} is the adjoint matrix-unit element $\bra{\hat{x}}E^\lambda_{T,S}\ket{x}$.
    Next, we use the normal-dual Clebsch--Gordan identity from Eq.~\eqref{eq:dual_normal_CG}.
    Using this equality we can transform bottom tensor network of dual CG tensor into a tensor network with normal CG tensors:
    \begin{equation}
        \includegraphics[width=0.3\textwidth, valign=c, page=9]{figures/tensor_networks.pdf} = \; \frac{1}{d_\lambda}  \includegraphics[width=0.3\textwidth, valign=c, page=10]{figures/tensor_networks.pdf} = \; \frac{1}{d_\lambda} \bra{y} E^\lambda_{S,T}\ket{\hat{y}}.
    \end{equation}
    Multiplying this element with the $x,\hat{x}$ element above gives the left-hand side of \cref{eq:thm_main_2}.
\end{proof}

Finally, we argue that essentially the same proof with minor modifications also holds for oracles $\cO,\tO,\iO$ from \cref{fig:4_types}.

\begin{corollary}\label{thm:app_all_query_oracles}
Let $\alpha=(\alpha_1,\ldots,\alpha_n)\in\{\mathrm{f},\mathrm{c},\mathrm{t},\mathrm{i}\}^n$ be an arbitrary sequence of query types.
For notational convenience define the four local coefficient matrices
\begin{align}
    \rho_{\mathrm{f}}(g)&\coloneqq \rho(g),
    &
    \rho_{\mathrm{c}}(g)&\coloneqq \bar{\rho}(g),
    &
    \rho_{\mathrm{t}}(g)&\coloneqq \rho(g)\tp,
    &
    \rho_{\mathrm{i}}(g)&\coloneqq \rho(g^{-1}).
\end{align}
For $a \in\{\mathrm{f},\mathrm{c},\mathrm{t},\mathrm{i}\}$, let $a \mathrm{O}$ denote the corresponding oracle and write $a \mathrm{O}_{\bullet,(y,x)}\coloneqq (I\otimes\bra{y})a\mathrm{O}(I\otimes\ket{x})$.
Then the compressed oracles exactly reproduce the Haar moment for this mixed query sequence:
\begin{align}
    &\Tr_{\bullet}\sof*{
    \alpha_n\mathrm{O}_{\bullet,(y_n,x_n)}
    \dotsb
    \alpha_1\mathrm{O}_{\bullet,(y_1,x_1)}
    \proj{\varnothing}_{\bullet}
    \bigl(\alpha_1\mathrm{O}_{\bullet,(\hat{x}_1,\hat{y}_1)}\bigr)^\dagger
    \dotsb
    \bigl(\alpha_n\mathrm{O}_{\bullet,(\hat{x}_n,\hat{y}_n)}\bigr)^\dagger
    }
    \nonumber\\
    &\qquad =
    \int_{\mathrm{Haar}} dg\,
    \bra{y_n}\rho_{\alpha_n}(g)\ket{x_n}
    \dotsb
    \bra{y_1}\rho_{\alpha_1}(g)\ket{x_1}
    \bra{\hat{x}_1}\rho_{\alpha_1}(g)^\dagger\ket{\hat{y}_1}
    \dotsb
    \bra{\hat{x}_n}\rho_{\alpha_n}(g)^\dagger\ket{\hat{y}_n}.
\end{align}
\end{corollary}
\begin{proof}
The case $\alpha=(\mathrm{f},\ldots,\mathrm{f})$ is the forward-query result of \cref{thm:app_forward_oracle} combined with \cref{lemma:mat_units}.
We show that changing any individual query type only changes the local matrix element in the same tensor-network proof.
For a transpose query at step $i$, the local factor satisfies
\begin{equation}
    \rho_{\mathrm{t}}(g)_{y_i,x_i}
    =
    \rho(g)\tp_{y_i,x_i}
    =
    \rho(g)_{x_i,y_i}.
\end{equation}
Thus the oracle $\tO$ is obtained from the forward oracle tensor network by exchanging the input and output wires at that step, while the matrix-unit contraction is unchanged.
For a conjugate query, the local representation species is $\bar{\rho}$ instead of $\rho$.
For an inverse query, unitarity gives
\begin{equation}
    \rho_{\mathrm{i}}(g)_{y_i,x_i}
    =
    \rho(g^{-1})_{y_i,x_i}
    =
    \bar{\rho}(g)_{x_i,y_i}.
\end{equation}
Hence $\iO$ is obtained by combining the conjugate replacement with the same wire exchange used for $\tO$.

Because these replacements are local in the query index, they can be performed independently for all entries of $\alpha$.
After the wire exchanges, the proof involves only a tensor product of copies of $\rho$ and $\bar{\rho}$.
Let $p$ be the number of forward or transpose factors and $q=n-p$ the number of conjugate or inverse factors.
The relevant representation space is then, up to the permutation of tensor factors dictated by $\alpha$,
\begin{align}
    V_\alpha
    \coloneqq
    V^{\otimes p}\otimes\bar V^{\otimes q}.
\end{align}
By complete reducibility for compact groups,
\begin{align}
    V_\alpha
    \simeq
    \bigoplus_{\lambda\in\hat{G}_\alpha}V_\lambda\otimes M_{\lambda,\alpha},
    \qquad
    \mathcal A_\alpha
    \coloneqq
    \End_G(V_\alpha).
\end{align}
Choose matrix units $E^{\lambda,\alpha}_{T,S}$ for the commutant algebra $\mathcal A_\alpha$.
The same Peter--Weyl orthogonality argument as in \cref{lemma:mat_units} applies with these matrix units and gives the displayed Haar integral.
In the defining representation of $\U(d)$ this commutant is the image of the walled Brauer algebra, possibly with a kernel outside the stable range, but the abstract compact-group statement uses $\End_G(V_\alpha)$.
\end{proof}

\section{\texorpdfstring{Exact path-recording oracles}{Exact path-recording oracle}}
\label{sec:exact_path_recording_oracle}

The oracles $\fO,\cO,\tO,\iO$ from \cref{fig:4_types} constructed in this work are Fourier compressed oracles: their auxiliary registers correspond to irreps of the group $G$, and one query is implemented by Clebsch--Gordan transforms, which are natively defined on the Fourier domain of the group $G$.
This is the natural basis for exactness as we showed in the previous section, but it is not the basis usually used in cryptographic compressed-oracle arguments.
In the random-function setting, Zhandry's compressed oracle records a small database of queried input-output pairs, and this database basis supports membership projectors and ``fundamental lemma'' arguments~\cite{zhandry2019record}.
Path-recording oracle for Haar-random unitaries uses the same database intuition and used to show the adaptive-security proof of pseudorandom unitaries~\cite{ma2024}.
For random permutations, database-basis compressed oracles are central because the ideal recording label is a partial injection~\cite{unruh2023towards,majenz2025permutation,carolan2025compressed}.

We therefore ask how our exact Fourier oracle looks after changing the auxiliary register to the database basis.
This yields exact path-recording oracle for arbitrary compact group $G$.
The key observation is that one can achieve this database basis with the help of Schur transforms.
Next, we work out the forward $\U(d)$ query example in detail and return to mixed query types at the end.

\subsection{\texorpdfstring{Truncated Peter--Weyl transform for $\U(d)$}{Truncated Peter-Weyl transform for U(d)}}
\label{sec:truncated_peter_weyl_transform_ud}

The two ancilla registers of our Fourier oracles in \cref{fig:4_types} naturally appear in the context of Peter--Weyl theorem.
It says that for a compact group $G$, matrix coefficients of irreducible representations form an orthogonal basis of $L^2(G)$, or equivalently
\begin{align}
    L^2(G)
    \simeq
    \bigoplus_{\lambda\in\widehat{G}}
    V_\lambda\otimes\overline{V}_\lambda.
\end{align}
A Fourier basis vector therefore carries one irrep label $\lambda$ and two basis labels, one in $V_\lambda$ and one in $\overline{V}_\lambda$.
This is the representation theoretic origin of the memory states $\ket{\lambda,M, \bar M'}$ used by our compressed oracle.

The finite-dimensional oracle only needs the part of this decomposition visible to the matrix coefficients generated by the allowed queries.
For the defining representation $\rho$ of $\U(d)$, degree-$t$ matrix coefficients are products the corresponding matrix entries:
\begin{align}
\label{eq:peter_weyl_ud_matrix_monomial}
    \bra{y} \rho(U)^{\otimes t} \ket{x} = U_{y_1,x_1}\cdots U_{y_t,x_t}.
\end{align}
The space of polynomials of degree $t$ in these variables can be identified with the symmetric subspace of local dimension $d^2$:
\begin{align}
\label{eq:peter_weyl_cauchy_space}
    \operatorname{Sym}^t\bigl(\mathbb{C}^d\otimes\overline{\mathbb{C}^d}\bigr)
    =
    \bigl((\mathbb{C}^d\otimes\overline{\mathbb{C}^d})^{\otimes t}\bigr)^{\mathfrak S_t},
\end{align}
where $\mathfrak S_t$ permutes the input-output pairs $(x,y)$ simultaneously.
This space can be identified with the direct sum over $U(d)$ irreps:
\begin{align}
\label{eq:peter_weyl_cauchy_identity}
    \operatorname{Sym}^t\bigl(\mathbb{C}^d\otimes\overline{\mathbb{C}^d}\bigr)
    \simeq
    \bigoplus_{\lambda\vdash_d t}
    V_\lambda\otimes \overline{V_\lambda}.
\end{align}
The way to see that is the following. After separating the input and output alphabets, Schur--Weyl duality gives
\begin{align}
\label{eq:peter_weyl_separated_schur}
    \bigl(\mathbb{C}^{d}\bigr)^{\otimes t}\otimes\bigl(\overline{\mathbb{C}^{d}}\bigr)^{\otimes t}
    \simeq
    \bigoplus_{\lambda,\mu\vdash_d t}
    V_\lambda\otimes \overline{V_\mu}
    \otimes M_\lambda\otimes M_\mu, 
\end{align}
Here $M_\lambda$ is the Specht multiplicity space for $\lambda$, with $m_\lambda\coloneqq\dim M_\lambda=|\SYT(\lambda)|$.
We also use that symmetric-group irreps are self-dual.
Moreover, it follows from the Schur's lemma that
\begin{align}
\label{eq:peter_weyl_specht_invariant}
    \bigl(M_\lambda\otimes M_\mu\bigr)^{\mathfrak S_t}
    \simeq
    \begin{cases}
        \mathbb{C}\ket{\Omega_\lambda}, & \lambda=\mu,\\
        0, & \lambda\neq \mu,
    \end{cases}
    \qquad
    \ket{\Omega_\lambda}
    \defeq
    \frac{1}{\sqrt{m_\lambda}}
    \sum_{T\in\SYT(\lambda)}\ket{T}\ket{T}.
\end{align}
Now let $W_t$ be an isometry which creates these maximally entangled spaces in every irrep sector $\lambda$:
\begin{align}
\label{eq:peter_weyl_W_isometry}
    W_t\ket{\lambda,L,\bar L'}
    \defeq
    \ket{\lambda,L,\bar L'}\ket{\Omega_\lambda}.
\end{align}
Then we define the truncated Peter--Weyl transform as
\begin{align}
\label{eq:peter_weyl_transform}
    \UPW
    \coloneqq
    \USch^{(d^2,t)}
    \bigl(\USch^{(d,t)}\otimes\UdSch^{(d,t)}\bigr)^\dagger
    W_t.
\end{align}

\subsection{\texorpdfstring{Exact Fourier and path-recording compressed forward oracles}{Exact Fourier and path-recording compressed forward oracles}}
\label{sec:fourier_compressed_forward_oracle}

Focusing on forward queries of defining representation of $\U(d)$, we write the Fourier ancilla basis as $\ket{\lambda,M,\bar{M}'}$, where $M$ and $M'$ are Gelfand--Tsetlin labels in the two Peter--Weyl registers.
The bar on $\bar{M}'$ is only a register label.
For $L\in\GT(\lambda)$, $M\in\GT(\mu)$, and $\lambda\in\mu+\square$, set
\begin{align}
    \CG^L_{M,x}
    \coloneqq
    \bra{L}\CG\ket{M,x},
    \qquad
    \dCG^M_{L,x}
    \coloneqq
    \sqrt{\frac{d_\mu}{d_\lambda}}\,
    \overline{\CG^L_{M,x}} .
\end{align}
On the Fourier basis states our forward oracle acts as
\begin{align}
\label{eq:appF_fourier_oracle_action}
\fO:\ket{x_t}\ket{\mu,M,\bar{M}'}
\longmapsto
\sum_{y_t\in[d]}
\ket{y_t}
\sum_{\lambda\in\mu+\square}
\sum_{L,L'\in\GT(\lambda)}
\CG^L_{M,x_t}\,
\dCG^{M'}_{L',y_t}\,
\ket{\lambda,L,\bar{L}'}.
\end{align}
The label ``Fourier'' only indicates the basis type of the auxiliary register. 

We can choose another natural basis, which we call database basis.
A length-$t$ database is an unordered multiset
\begin{align}
    D^{(t)}
    =
    \{(x_1,y_1),\dotsc,(x_t,y_t)\},
\end{align}
where $x=(x_1,\dotsc,x_{t-1},x_t)$ and $y=(y_1,\dotsc,y_{t-1},y_t)$ with all $x_i,y_j \in [d]$.
Let's define
\begin{align}
    w^{(t)}_{ab}
    &\coloneqq
    \#\{k\in[t]:(x_k,y_k)=(a,b)\},
    \\
    \binom{t}{w^{(t)}}
    &\coloneqq
    \frac{t!}{\prod_{a,b=1}^{d}w^{(t)}_{ab}!},
    \\
    Y_{w^{(t)}}
    &\cong
    \prod_{a,b=1}^{d} \mathfrak{S}_{w^{(t)}_{ab}}.
\end{align}
Then the database state written in the computational Hilbert space is normalized orbit state under Young subgroup $Y_{w^{(t)}}$:
\begin{align}
\label{eq:appF_orbit_state}
    \ket{D^{(t)}_{\mathrm{comp}}}
    \coloneqq
    \sqrt{\frac{1}{\binom{t}{w^{(t)}}}}
    \sum_{\pi\in \mathfrak S_t/Y_{w^{(t)}}}
    \ket{(x_{\pi^{-1}(1)},y_{\pi^{-1}(1)}),\dotsc,(x_{\pi^{-1}(t)},y_{\pi^{-1}(t)})}.
\end{align}
We can use the pair-alphabet Schur transform to define the compressed database basis vector as
\begin{align}
\label{eq:appF_database_state}
    \ket{D^{(t)}}
    \coloneqq
    \USch^{(d^2,t)}\ket{D^{(t)}_{\mathrm{comp}}}.
\end{align}

We now relate this basis to the Fourier basis of Eq.~\eqref{eq:appF_fourier_oracle_action}.
For $T\in\SYT(\lambda)$, let $T^i$ be the subtableau containing $1,\dotsc,i$, and let $T_i$ denote its shape.
The ordinary Schur coefficients are defined recursively by
\begin{align}
\label{eq:appF_sch_recursion}
    \Sch^x_{\lambda,L,T}
    &\coloneqq
    \sum_{M\in\GT(T_{t-1})}
    \CG^L_{M,x_t}\,
    \Sch^{x_{<t}}_{T_{t-1},M,T^{t-1}}.
\end{align}
The dual Schur coefficients satisfy
\begin{align}
\label{eq:appF_dsch_relation}
    \dSch^y_{\lambda,L,T}
    =
    \sqrt{\frac{1}{d_\lambda}}\,
    \overline{\Sch^y_{\lambda,L,T}}.
\end{align}
We write $\UdSch^{(d,t)}$ for the corresponding dual Schur transform:
\begin{align}
    \UdSch^{(d,t)}\ket{y_1,\dotsc,y_t}
    =
    \sum_{\lambda\vdash_d t}
    \sum_{L\in\GT(\lambda)}
    \sum_{T\in\SYT(\lambda)}
    \sqrt{d_\lambda}\,
    \dSch^y_{\lambda,L,T}
    \ket{\bar{\lambda},\bar L,T}.
\end{align}
Equivalently, $\UdSch^{(d,t)}=\bUSch^{(d,t)}$, but the notation $\UdSch$ keeps the dual-Clebsch--Gordan normalization visible.
Then one can calculate
\begin{align}
\label{eq:appF_V_on_D}
    \UPW^\dagger\ket{D^{(t)}}_{\mathrm{DB}}
    &=
    W_t^\dagger
    \bigl(\USch^{(d,t)}\otimes\UdSch^{(d,t)}\bigr)
    \bigl(\USch^{(d^2,t)}\bigr)^\dagger
    \ket{D^{(t)}}_{\mathrm{DB}}
    \nonumber\\
    &=
    \sqrt{\frac{1}{\binom{t}{w^{(t)}}}}
    \sum_{\pi\in \mathfrak S_t/Y_{w^{(t)}}}
    W_t^\dagger
    \bigl(\USch^{(d,t)}\otimes\UdSch^{(d,t)}\bigr)
    \ket{x_{\pi^{-1}(1)},\dotsc,x_{\pi^{-1}(t)}}
    \ket{y_{\pi^{-1}(1)},\dotsc,y_{\pi^{-1}(t)}}
    \nonumber\\
    &=
    \sqrt{\binom{t}{w^{(t)}}}
    W_t^\dagger
    \bigl(\USch^{(d,t)}\otimes\UdSch^{(d,t)}\bigr)
    \ket{x_1,\dotsc,x_t}
    \ket{y_1,\dotsc,y_t}
    \nonumber\\
    &=
    \sqrt{\binom{t}{w^{(t)}}}
    W_t^\dagger
    \sum_{\lambda,\mu\vdash_d t}
    \sum_{\substack{L\in\GT(\lambda),\,L'\in\GT(\mu)\\T\in\SYT(\lambda),\,S\in\SYT(\mu)}}
    \Sch^x_{\lambda,L,T}
    \sqrt{d_\mu}\,
    \dSch^y_{\mu,L',S}
    \ket{\lambda,L,T}
    \ket{\bar{\mu},\bar{L}',S}
    \nonumber\\
    &=
    \sqrt{\binom{t}{w^{(t)}}}
    \sum_{\lambda\vdash_d t}
    \sum_{T\in\SYT(\lambda)}
    \frac{1}{\sqrt{m_\lambda}}
    \sum_{L,L'\in\GT(\lambda)}
    \Sch^x_{\lambda,L,T}
    \sqrt{d_\lambda}\,
    \dSch^y_{\lambda,L',T}
    \ket{\lambda,L,\bar{L}'}
    \nonumber\\
    &=
    \sqrt{\binom{t}{w^{(t)}}}
    \sum_{\lambda\vdash_d t}
    \sum_{T\in\SYT(\lambda)}
    \frac{1}{\sqrt{m_\lambda}}
    \sum_{L,L'\in\GT(\lambda)}
    \Sch^x_{\lambda,L,T}
    \overline{\Sch^y_{\lambda,L',T}}
    \ket{\lambda,L,\bar{L}'}.
\end{align}
We define the forward oracle in the database basis as
\begin{align}
\label{eq:appF_database_oracle}
    \fO_{\mathrm{DB}}
    \coloneqq
    \UPW \fO \UPW^\dagger.
\end{align}
It acts as follows. Choose a length-$(t-1)$ input database
\begin{align}
    D^{(t-1)}
    &=
    \{(x_1,y_1),\dotsc,(x_{t-1},y_{t-1})\},
\end{align}
with weight $w^{(t-1)}$.
For an output database with weight $\tilde{w}^{(t)}$, write
\begin{align}
    \widetilde D^{(t)}
    =
    \{(\tilde{x}_1,\tilde{y}_1),\dotsc,(\tilde{x}_t,\tilde{y}_t)\},
\end{align}
where $\tilde{x}=(\tilde{x}_1,\dotsc,\tilde{x}_t)$ and $\tilde{y}=(\tilde{y}_1,\dotsc,\tilde{y}_t)$.
Define
\begin{align}
\label{eq:appF_K_def}
    K_\lambda(\tilde{x},x;T,S)
    \coloneqq
    \sum_{L\in\GT(\lambda)}
    \overline{\Sch^{\tilde{x}}_{\lambda,L,T}}\,
    \Sch^x_{\lambda,L,S}.
\end{align}
This is a symmetric-group matrix-unit coefficient:
\begin{align}
\label{eq:appF_K_matrix_unit}
    K_\lambda(\tilde{x},x;T,S)
    =
    \bra{\tilde{x}}E^\lambda_{T,S}\ket{x},
    \qquad
    \USch^{(d,t)}E^\lambda_{T,S}\bigl(\USch^{(d,t)}\bigr)^\dagger
    =
    \1_{d_\lambda}\otimes\ketbra{T}{S}.
\end{align}
Then the exact path-recording oracle action is obtained by inserting the database-basis resolution of the identity after applying $\UPW \fO \UPW^\dagger$.
In the first equality below we use the last equality of Eq.~\eqref{eq:appF_V_on_D} for the length-$(t-1)$ input database and then apply the Fourier update Eq.~\eqref{eq:appF_fourier_oracle_action}.
\begin{align}
\fO_{\mathrm{DB}}\ket{x_t}\ket{D^{(t-1)}}
&=
\sum_{y_t\in[d]}
\sum_{\widetilde D^{(t)}}
\ket{y_t}\ket{\widetilde D^{(t)}}
\bra{\widetilde D^{(t)}}\bra{y_t}
\UPW \fO \UPW^\dagger
\ket{x_t}\ket{D^{(t-1)}}
\nonumber\\
\label{eq:appF_exact_database_oracle}
&=
\sum_{y_t\in[d]}\ket{y_t}
\sum_{\widetilde D^{(t)}}
A_t\bigl(\widetilde D^{(t)};D^{(t-1)},x_t,y_t\bigr)
\ket{\widetilde D^{(t)}},
\\
\label{eq:appF_exact_amplitude}
A_t\bigl(\widetilde D^{(t)};D^{(t-1)},x_t,y_t\bigr)
&\coloneqq
\bra{\widetilde D^{(t)}}\bra{y_t}
\UPW \fO \UPW^\dagger
\ket{x_t}\ket{D^{(t-1)}}
\\
&=
\sqrt{\binom{t-1}{w^{(t-1)}}}
\sum_{\lambda\vdash_d t}
\sum_{S\in\SYT(\lambda)}
\sqrt{\frac{d_{S_{t-1}}}{m_{S_{t-1}}d_\lambda}}
\sum_{L,L'\in\GT(\lambda)}
\Sch^x_{\lambda,L,S}
\overline{\Sch^y_{\lambda,L',S}}
\bra{\widetilde D^{(t)}}\UPW\ket{\lambda,L,\bar L'}
\nonumber\\
&=
\sqrt{\binom{t}{\tilde{w}^{(t)}}\binom{t-1}{w^{(t-1)}}}
\sum_{\lambda\vdash_d t}
\sum_{T,S\in\SYT(\lambda)}
\sqrt{\frac{d_{S_{t-1}}}{m_{S_{t-1}}d_\lambda m_\lambda}}
\nonumber\\
&\qquad{}\times
\sum_{L,L'\in\GT(\lambda)}
\overline{\Sch^{\tilde{x}}_{\lambda,L,T}}\,
\Sch^{\tilde{y}}_{\lambda,L',T}\,
\Sch^x_{\lambda,L,S}
\overline{\Sch^y_{\lambda,L',S}}
\nonumber\\
&=
\sqrt{\binom{t}{\tilde{w}^{(t)}}\binom{t-1}{w^{(t-1)}}}
\sum_{\lambda\vdash_d t}
\sum_{T,S\in\SYT(\lambda)}
\sqrt{\frac{d_{S_{t-1}}}{m_{S_{t-1}}d_\lambda m_\lambda}}\,
K_\lambda(\tilde{x},x;T,S)
\overline{K_\lambda(\tilde{y},y;T,S)}.
\nonumber
\end{align}
Here $S_{t-1}$ is obtained from $S$ by removing the box containing $t$.

For $D=\{(a_k,b_k)\}_{k=1}^{t}$ define the marginals
\begin{align}
    \margx(D)_a
    &\coloneqq
    \#\{k:a_k=a\},
    &
    \margy(D)_b
    &\coloneqq
    \#\{k:b_k=b\}.
\end{align}
Since $E^\lambda_{T,S}\in\C[\mathfrak S_t]$, nonzero amplitude implies
\begin{align}
\label{eq:appF_support}
    A_t\bigl(\widetilde D^{(t)};D^{(t-1)},x_t,y_t\bigr)\neq0
    \quad\Longrightarrow\quad
    \begin{cases}
    \margx\bigl(\widetilde D^{(t)}\bigr)
    =
    \margx\bigl(D^{(t-1)}\bigr)+\e_{x_t},\\
    \margy\bigl(\widetilde D^{(t)}\bigr)
    =
    \margy\bigl(D^{(t-1)}\bigr)+\e_{y_t}.
    \end{cases}
\end{align}
So the exact oracle records the updated input and output multisets, while the pairing between them can be coherently reshuffled.

\begin{example}
For $t=1$, the input database is $D^{(0)}=\varnothing$ and the only
partition is $\lambda=(1)$. The unique tableau has
$S_0=\varnothing$, with
$d_{\varnothing}=m_{\varnothing}=1$, while
$d_{(1)}=d$, $m_{(1)}=1$, and
\begin{equation}
K_{(1)}(\tilde{x},x_1)
=
\delta_{\tilde{x},x_1}.
\end{equation}
Consequently,
\begin{align}
A_1\bigl(\{(x_1,y_1)\};\varnothing,x_1,y_1\bigr)
&=
\frac{1}{\sqrt d},
\qquad \qquad
\fO_{\mathrm{DB}}\ket{x_1}\ket{\varnothing}
=
\frac{1}{\sqrt d}
\sum_{y_1\in[d]}
\ket{y_1}\ket{\{(x_1,y_1)\}}.
\label{eq:appF_t1_action}
\end{align}
Thus the first query is exactly flat over the output alphabet. The coefficient in Eq.~\eqref{eq:appF_exact_amplitude} is already
nontrivial at $t=2$. 

For $t=2$, let
\begin{equation}
D^{(1)}=\{(x_1,y_1)\},
\end{equation}
and fix the second-query input $x_2$ and one value $y_2$ of the
output register. Assume first that $d\geq 2$, $x_2\neq x_1$, and
$y_2\neq y_1$. Each of the two $\mathfrak{S}_2$ sectors has a unique standard
tableau, and its matrix unit is
\begin{align}
E^{(2)}_{T,T}
&=
\frac{1}{2}\bigl(\1+(12)\bigr),
&
E^{(1,1)}_{T,T}
&=
\frac{1}{2}\bigl(\1-(12)\bigr),
\label{eq:appF_S2_projectors}
\\
d_{(2)}
&=
\frac{d(d+1)}{2},
&
d_{(1,1)}
&=
\frac{d(d-1)}{2}.
\label{eq:appF_S2_dimensions}
\end{align}
Moreover, $m_{(2)}=m_{(1,1)}=1$.

The support rule in Eq.~\eqref{eq:appF_support} leaves exactly two
output databases. For the direct pairing, choose
\begin{equation}
\tilde{x}=(x_1,x_2),
\qquad
\tilde{y}=(y_1,y_2),
\end{equation}
whereas for the crossed pairing, choose
\begin{equation}
\tilde{x}=(x_1,x_2),
\qquad
\tilde{y}=(y_2,y_1).
\end{equation}
Suppressing the unique tableau label in each sector,
Eq.~\eqref{eq:appF_K_matrix_unit} gives
\begin{align}
K_{(2)}\bigl((x_1,x_2),(x_1,x_2)\bigr)
&=
\frac{1}{2},
&
K_{(1,1)}\bigl((x_1,x_2),(x_1,x_2)\bigr)
&=
\frac{1}{2},
\\
K_{(2)}\bigl((y_1,y_2),(y_1,y_2)\bigr)
&=
\frac{1}{2},
&
K_{(1,1)}\bigl((y_1,y_2),(y_1,y_2)\bigr)
&=
\frac{1}{2},
\\
K_{(2)}\bigl((y_2,y_1),(y_1,y_2)\bigr)
&=
\frac{1}{2},
&
K_{(1,1)}\bigl((y_2,y_1),(y_1,y_2)\bigr)
&=
-\frac{1}{2}.
\end{align}
Since the two pairs in either output database are distinct, $\binom{2}{\tilde{w}^{(2)}}=2$.
Substitution into Eq.~\eqref{eq:appF_exact_amplitude} yields
\begin{align}
&A_2\Bigl(
\{(x_1,y_1),(x_2,y_2)\};
\{(x_1,y_1)\},x_2,y_2
\Bigr)
=
\frac{1}{2}
\left(
\frac{1}{\sqrt{d+1}}
+
\frac{1}{\sqrt{d-1}}
\right)
=
\frac{1}{\sqrt d}
+
\frac{3}{8d^{5/2}}
+
O\!\left(d^{-9/2}\right),
\label{eq:appF_t2_direct}
\\[1mm]
&A_2\Bigl(
\{(x_1,y_2),(x_2,y_1)\};
\{(x_1,y_1)\},x_2,y_2
\Bigr)
=
\frac{1}{2}
\left(
\frac{1}{\sqrt{d+1}}
-
\frac{1}{\sqrt{d-1}}
\right)
=
-\frac{1}{2d^{3/2}}
-
\frac{5}{16d^{7/2}}
+
O\!\left(d^{-11/2}\right).
\label{eq:appF_t2_crossed}
\end{align}
All other output databases have zero amplitude. Thus the direct
append branch is of order $d^{-1/2}$, whereas the first coherent
re-pairing is smaller by one power of $d$. Finally, if
$x_2\neq x_1$ but $y_2=y_1$, then
\begin{align}
&A_2\Bigl(
\{(x_1,y_1),(x_2,y_1)\};
\{(x_1,y_1)\},x_2,y_1
\Bigr)=
\frac{1}{\sqrt{d+1}}.
\label{eq:appF_t2_output_collision}
\end{align}
If $x_2=x_1$ but $y_2\neq y_1$, then
\begin{align}
&A_2\Bigl(
\{(x_1,y_1),(x_1,y_2)\};
\{(x_1,y_1)\},x_1,y_2
\Bigr)=
\frac{1}{\sqrt{d+1}}.
\label{eq:appF_t2_input_collision}
\end{align}
If $x_2=x_1$ and $y_2=y_1$, then
\begin{align}
&A_2\Bigl(
\{(x_1,y_1),(x_1,y_1)\};
\{(x_1,y_1)\},x_1,y_1
\Bigr)=
\sqrt{\frac{2}{d+1}}.
\label{eq:appF_t2_double_collision}
\end{align}
As a check, for $x_2\neq x_1$ the total squared norm of the $t=2$
output is
\begin{equation}
(d-1)\frac{1}{2}
\left(
\frac{1}{d+1}
+
\frac{1}{d-1}
\right)
+
\frac{1}{d+1}
=
1.
\label{eq:appF_t2_norm_check}
\end{equation}
When $x_2=x_1$, the squared norm is
\begin{equation}
(d-1)\frac{1}{d+1}
+
\frac{2}{d+1}
=
1,
\end{equation}
since the $d-1$ branches with $y_2\neq y_1$ have amplitude
$1/\sqrt{d+1}$, while the repeated branch $y_2=y_1$ has amplitude
$\sqrt{2/(d+1)}$.
\end{example}

Finally, we compare our exact forward oracle with Ma--Huang approximate path-recording oracle.
Ma and Huang~\cite{ma2024} give two related path-recording constructions: a one-register oracle for forward-only access and a two-register oracle for forward and inverse access.  
Their forward oracle keeps only the ``append-only'' $D \to D\cup\{(x,y)\}$ transition:
\begin{align}
\label{eq:appF_MH_oracle}
    \fO_{\mathrm{MH}}\ket{x}\ket{D}
    =
    \frac{1}{\sqrt{d-|D|}}
    \sum_{\substack{y\in[d]\\ y\notin\operatorname{Im}(D)}}
    \ket{y}\ket{D\cup\{(x,y)\}},
\end{align}
on relation databases with $|\operatorname{Im}(D)|=|D|<d$.

This oracle approximates our exact oracle.
For a $t$-query algorithm $\mathcal{A}$, let $\rho_{\mathrm{exact}}(\mathcal{A})$ be the final reduced state produced by $\fO_{\mathrm{DB}}$, and let $\rho_{\mathrm{MH}}(\mathcal{A})$ be the final reduced state produced by Eq.~\eqref{eq:appF_MH_oracle}.
Since our oracle is exact, it can be replaced by the corresponding Haar average. 
To that end, Ma and Huang proved in \cite{ma2024} that
\begin{align}
\label{eq:appF_MH_approx}
    \operatorname{TD}\bigl(
    \rho_{\mathrm{exact}}(\mathcal{A}),
    \rho_{\mathrm{MH}}(\mathcal{A})
    \bigr)
    \leq 
    O(t^2/d),
\end{align}
which is negligible when $t$ is small enough compared to $d$.

\subsection{\texorpdfstring{Mixed queries}{Mixed queries}}
\label{sec:mixed_queries_path_recording}

For mixed queries it is better to use signed relation records rather than a single pair alphabet.
Note that forward and transpose queries correspond to $\rho$, while conjugate and inverse queries use $\bar\rho$.
Transpose and inverse queries only exchange the local matrix indices. The four local matrix coefficients are:
\begin{align}
    \mathrm{f}:&\quad \rho(g)_{y,x},
    &
    \mathrm{t}:&\quad \rho(g)_{x,y},
    &
    \mathrm{c}:&\quad \bar\rho(g)_{y,x},
    &
    \mathrm{i}:&\quad \bar\rho(g)_{x,y}.
\end{align}
A mixed database is therefore a pair $D=(L,R)$ of multisets. 
The multiset $L$ stores positive records and $R$ stores negative records, with convention
\begin{align}
\label{eq:appF_mixed_record_convention}
    \mathrm{f}: (x,y)\in L,
    \qquad
    \mathrm{t}: (y,x)\in L,
    \qquad
    \mathrm{c}: (x,y)\in R,
    \qquad
    \mathrm{i}: (y,x)\in R.
\end{align}
The coefficient function corresponding to \cref{eq:peter_weyl_ud_matrix_monomial} represented by $D=(L,R)$ is
\begin{align}
\label{eq:appF_signed_relation_monomial}
    \prod_{(x,y)\in L}\rho(g)_{y,x}
    \prod_{(x,y)\in R}\overline{\rho(g)_{y,x}}.
\end{align}
The implementation of mixed queries relies on mixed CG transforms, which we outline in \hyperref[app:mixed_cg_transforms]{Appendix~D.1}.

\paragraph{The Ma--Huang approach.}
Use the pair $D=(L,R)$ introduced above, and write $\ket{L,R}$ for the two orthonormal relation-state registers.  In the notation of this manuscript, the left and right partial isometries of Ref.~\cite{ma2024} are
\begin{align}
\label{eq:appF_MH_VL_VR}
    V_{\mathrm{MH}}^L\ket{x}\ket{L,R}
    &\coloneqq
    \frac{1}{\sqrt{d-|\operatorname{Im}(L\cup R)|}}
    \sum_{\substack{y\in[d]\\y\notin\operatorname{Im}(L\cup R)}}
    \ket{y}\ket{L\cup\{(x,y)\},R},
    \\
    V_{\mathrm{MH}}^R\ket{y}\ket{L,R}
    &\coloneqq
    \frac{1}{\sqrt{d-|\operatorname{Dom}(L\cup R)|}}
    \sum_{\substack{x\in[d]\\x\notin\operatorname{Dom}(L\cup R)}}
    \ket{x}\ket{L,R\cup\{(x,y)\}}.
\end{align}
These formulas apply when the combined database length is at most $d-1$.  The first map creates a record in $L$, while the second creates a record in $R$.  Their strong forward path-recording oracle is not simply their sum, because the creation and annihilation branches would overlap.  It is the partial isometry
\begin{align}
\label{eq:appF_MH_strong_oracle}
    V_{\mathrm{MH}}
    \coloneqq
    V_{\mathrm{MH}}^L
    \bigl(I-V_{\mathrm{MH}}^R(V_{\mathrm{MH}}^R)^\dagger\bigr)
    +
    \bigl(I-V_{\mathrm{MH}}^L(V_{\mathrm{MH}}^L)^\dagger\bigr)
    (V_{\mathrm{MH}}^R)^\dagger.
\end{align}
Thus a forward query either creates an $L$-record through $V_{\mathrm{MH}}^L$ or annihilates an $R$-record through $(V_{\mathrm{MH}}^R)^\dagger$.  The two projectors remove the overlap of these branches.  Taking the adjoint gives the inverse query, which either creates an $R$-record through $V_{\mathrm{MH}}^R$ or annihilates an $L$-record through $(V_{\mathrm{MH}}^L)^\dagger$.

\paragraph{Fourier interpretation of $L$ and $R$.}
For $\U(d)$, a rational highest weight (staircase) is described by its positive and negative diagrams $[\alpha,\beta]$.  The mixed Pieri rules have the schematic form
\begin{align}
\label{eq:appF_MH_fourier_branches}
    [\alpha,\beta]\otimes\square
    &\longrightarrow
    [\alpha+\square,\beta]
    \quad\text{or}\quad
    [\alpha,\beta-\square],
    \\
    [\alpha,\beta]\otimes\bar\square
    &\longrightarrow
    [\alpha,\beta+\square]
    \quad\text{or}\quad
    [\alpha-\square,\beta],
\end{align}
where only valid addable or removable boxes occur.  If $p$ defining and $q$ dual defining factors have been processed, a branch with $s$ contractions has $|\alpha|=p-s$ and $|\beta|=q-s$; deleting a record from the opposite Ma--Huang register is the path-recording version of the same contraction.  Equation~\eqref{eq:appF_MH_fourier_branches} identifies the representation-theoretic role of the two Ma--Huang registers.  The $L$ register refines the positive diagram $\alpha$, and the $R$ register refines the negative diagram $\beta$: $V_{\mathrm{MH}}^L$ corresponds to creating a positive box, $(V_{\mathrm{MH}}^R)^\dagger$ to removing a negative box, $V_{\mathrm{MH}}^R$ to creating a negative box, and $(V_{\mathrm{MH}}^L)^\dagger$ to removing a positive box.  Along a branch matched to a rational Pieri path and starting from the empty database, the gradings satisfy
\begin{align}
\label{eq:appF_MH_degree_correspondence}
    |L|=|\alpha|,
    \qquad
    |R|=|\beta|,
    \qquad
    |L|-|R|=|\alpha|-|\beta|.
\end{align}

Our Forier oracle automatically implements an update accroding to \cref{eq:appF_MH_fourier_branches}.
However, as was shown in \cref{sec:fourier_compressed_forward_oracle} its database action is much more complicated compared to \cref{eq:appF_MH_VL_VR}.

\subsection{\texorpdfstring{General exact Fourier and path-recording oracles}{General exact Fourier and path-recording oracles}}
\label{sec:general_fourier_recording_database}
We now return to an arbitrary compact group representation $\rho:G\to\U(V)$.
The previous appendix used unordered input-output pairs because $\rho=\square$ for $\U(d)$ has the Schur--Weyl commutant as quotient of $\C[\mathfrak S_t]$.
For general $(G,\rho)$, the replacement is invariance under the full commutant of $\rho^{\otimes t}$ (or $\rho^{\otimes p} \otimes \bar{\rho}^{\otimes q}$ when using mixed queries).
Set
\begin{align}
\label{eq:general_db_H_A_def}
    \mathcal A_t
    \coloneqq
    \End_G(V^{\otimes t}).
\end{align}
Choose the Schur transform
\begin{align}
\label{eq:general_db_schur}
    \USch:
    V^{\otimes t}
    \longrightarrow
    \bigoplus_{\lambda\in\hat{G}^{(t)}}
    V_\lambda\otimes M_{\lambda}.
\end{align}
The map $\USch$ and the multiplicity spaces $M_\lambda$ depend on $t$, $G$, and $\rho$; we suppress this dependence in the notation.
The degree-$t$ Fourier memory is the corresponding visible Peter--Weyl space
\begin{align}
\label{eq:general_pw_visible_space}
    \mathcal{H}_{\mathrm{PW}}^{(t)}(G,\rho)
    \simeq
    \bigoplus_{\lambda\in\hat{G}^{(t)}}
    V_\lambda\otimes\bar V_\lambda .
\end{align}
Then
\begin{align}
\label{eq:general_db_commutants}
    \mathcal A_t
    &\simeq
    \bigoplus_{\lambda\in\hat{G}^{(t)}}
    \1_{V_\lambda}\otimes\End(M_{\lambda}),
    \\
    \mathcal A_t'
    &\simeq
    \bigoplus_{\lambda\in\hat{G}^{(t)}}
    \End(V_\lambda)\otimes\1_{M_{\lambda}},
\end{align}
where $\mathcal A_t'$ denotes the commutant of $\mathcal A_t$.
The degree-$t$ database space can be seen as subspace of $V^{\otimes t}\otimes\bar{V}^{\otimes t}$.
\begin{align}
\label{eq:general_db_code}
    \mathsf{DB}_t(G,\rho)
    \coloneqq
    \operatorname{span}\{\dket{X}:X\in \mathcal A_t'\}
    \subseteq
    V^{\otimes t}\otimes\bar{V}^{\otimes t}.
\end{align}
Equivalently,
\begin{align}
\label{eq:general_db_invariants}
    \mathsf{DB}_t(G,\rho)
    =
    \bigl(V^{\otimes t}\otimes\bar{V}^{\otimes t}\bigr)^{\mathcal U(\mathcal A_t)},
\end{align}
where $u\in\mathcal U(\mathcal A_t)$ acts as $u\otimes\bar u$.
Thus the database space is invariant under every unitary symmetry which commutes with $\rho^{\otimes t}$.

Let
\begin{align}
\label{eq:general_db_matrix_unit}
    F^\lambda_{a,b}
    \coloneqq
    \USch^\dagger
    \bigl(\ketbra{\lambda,a}{\lambda,b}\otimes\1_{M_{\lambda}}\bigr)
    \USch.
\end{align}
Since $\norm{F^\lambda_{a,b}}_{\mathrm{HS}}^2=\dim M_{\lambda}$, an orthonormal database basis is
\begin{align}
\label{eq:general_db_basis}
    \ket{D^{(t)}_{\lambda,a,b}}_{\mathrm{DB}}
    \coloneqq
    \frac{1}{\sqrt{\dim M_{\lambda}}}\,
    \dket{F^\lambda_{a,b}}.
\end{align}
Using two Schur transforms, we get
\begin{align}
\label{eq:general_db_two_schur_basis}
    J_t\ket{\lambda,a}\ket{\bar{\lambda},b}
    \coloneqq
    \ket{D^{(t)}_{\lambda,a,b}}_{\mathrm{DB}}
    =
    \bigl(\USch\otimes\bUSch\bigr)^\dagger
    \Bigl(
    \ket{\lambda,a}\ket{\bar{\lambda},b}
    \ket{\Omega_{\lambda}}
    \Bigr),
\end{align}
where $\ket{\Omega_{\lambda}}$ are EPR states on multiplicity registers:
\begin{align}
\label{eq:general_db_multiplicity_epr}
    \ket{\Omega_{\lambda}}
    =
    \frac{1}{\sqrt{\dim M_{\lambda}}}
    \sum_{r=1}^{\dim M_{\lambda}}
    \ket{r}_{M_{\lambda}}\ket{\bar r}_{\bar{M}_{\lambda}}.
\end{align}
This formula defines the isometry $J:\mathcal{H}_{\mathrm{PW}}^{(t)}(G,\rho)\to\mathsf{DB}_t(G,\rho)$.
If a basis of $V$ is fixed, the vectors in $V^{\otimes t}\otimes\bar V^{\otimes t}$ may be read as ordered pair records $\ket{(x_1,y_1),\ldots,(x_t,y_t)}$.
This two-Schur map $J$ gives the canonical coordinates of the Fourier pair $\ket{\lambda,a}\ket{\bar\lambda,b}$ inside this database space.

Unlike the defining $\U(d)$ case, there is no nice analogue of $\USch^{(d^2,t)}$ that identifies a simple pair-alphabet Schur basis with the whole visible Peter--Weyl truncation space $\mathcal{H}_{\mathrm{PW}}^{(t)}(G,\rho)$.
Thus for general $(G,\rho)$ we define the exact database space by the commutant-invariant code above.

The corresponding database-basis oracle is obtained by conjugating the Fourier oracle degree by degree.
For the forward query,
\begin{align}
\label{eq:general_db_oracle}
    \fO_{\mathrm{DB},t}
    =
    (J_t\otimes I_V)\,
    \fO_t\,
    (J_{t-1}^\dagger\otimes I_V).
\end{align}
The same formula applies to $\cO,\tO,\iO$ with the corresponding Fourier compressed oracles and Schur transforms.

\section{\texorpdfstring{Efficient Clebsch--Gordan transforms for $\U(d)$}{Efficient Clebsch-Gordan transforms for U(d)}}\label{app:cg_trans}

In this section, we provide a proof of \cref{thm:main_2} by describing an efficient construction of (dual) Clebsch--Gordan transforms for the unitary group with defining representation $\square$, which is needed to achieve $\mathrm{poly}(n)$ complexity, where $n=\log_2(d)$.
The main components of this construction were first described in detail explicitly in \cite{burchardt2025krovi}, based on the ideas from \cite{harrow2005applications,krovi2019efficient}.
We present them here with minor modifications and adaptations needed for our setting. 

We start with the $\mathrm{CG}$ transform, which is presented in \cref{fig:CG_circuit_comp}. 
It consists of two main components: preprocessing gate $\mathrm{P}$ and compressed $\widetilde{\mathrm{CG}}$ transform.
The preprocessing block changes the representation coordinates from $(p^{(k-1)},\widetilde{M}^{(k-1)},x_k)$ to $(p^{(k)},N^{(k-1)},\tilde{x}_k)$.
The circuit $\widetilde{\CG}_k$ then maps $(N^{(k-1)},\tilde{x}_k)$ to $(\widetilde{M}^{(k)},z_k)$ and leaves $p^{(k)}$ untouched.

\begin{figure*}[!h]
    \centering
    $\includegraphics[width=0.22\textwidth, valign=c, page=4]{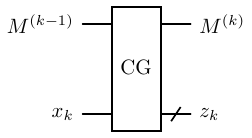}
    =
    \includegraphics[width=0.4\textwidth, valign=c, page=5]{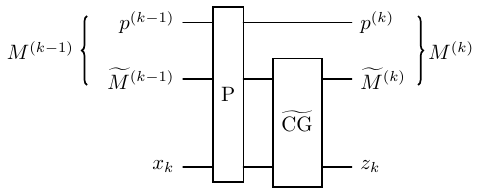}$
    \caption{Efficient $\mathrm{CG}$ transform consists of a preprocessing operation $\mathrm{P}$, which is needed to modify compressed Gelfand--Tsetlin pattern $\widetilde{M}^{(k-1)}$ upon arrival of new symbol $x_k$.}
    \label{fig:CG_circuit_comp}
\end{figure*}

The preprocessing gate $\mathrm{P}$ is needed to efficiently use the memory space by preparing the input $\ket{x_k}$ to be processed correctly within the compressed $\widetilde{\mathrm{CG}}$ transform.
The operation $\mathrm{P}$ consists of several steps, see \cref{fig:CG_circuit_P}.

The main intuition $\mathrm{P}$ comes from the following.
Basis vectors of unitary group irreps are labelled by Gelfand--Tsetlin patterns or, equivalently, by Semistandard Young tableaux.
Each pattern has an associated weight $w=(w_1,\dotsc,w_d)$, which counts number of $1$, $2$ and so on. It can be equivalently represented by a composition $\mu$ and a alphabet map $p$, that is, $w \cong (\mu,p)$, \cite{burchardt2025krovi}. 
A given Gelfand--Tsetlin pattern $M \in \mathrm{GT}(\lambda)$ can be equivalently represented by a smaller GT pattern $\widetilde{M} \in \mathrm{GT}(\lambda,\mu)$ of length $\ell(\mu)$, where $\ell(\mu)$ is the length of the composition $\mu$, together with the alphabet map $p$, that is, $M \cong (\widetilde{M},p)$. 
For example, a Gelfand--Tsetlin pattern $M = ((0),(2,0),(2,0,0),(2,1,0,0),(3,2,0,0,0))$ corresponds to $\widetilde{M}=((2),(2,1),(3,2,0))$ and $p = (2, 4, 5)$:
\begin{equation}
\sof*{\,
\begin{smallmatrix}
    3 & & 2 & & 0 & & 0 & & 0\\
    & 2 & & 1 & & 0 & & 0 \\
    & & 2 & & 0 & & 0 & \\
    & & & 2 & & 0 & & \\
    & & & & 0 & & & \\
\end{smallmatrix}} 
\quad \equiv \quad 
\of*{
\sof*{\begin{smallmatrix}
    3 & & 2 & & 0 \\
    & 2 & & 1 & \\
    & & 2 & & \\
\end{smallmatrix}}, \; (2, 4, 5)
}.
\end{equation}

Gate $\mathrm{P}$ effectively implements the above compression by handling newly arrived symbol $x_k \in [d]$ and updating old pair $(\widetilde{M},p)$. It consists of four steps $A$, $B$, $C$, and $D$.
We use the convention that $c_k=1$ means that $x_k$ is already present in $p^{(k-1)}$, while $c_k=0$ means that $x_k$ is a new alphabet symbol.

\begin{figure*}[!h]
    \centering
    \includegraphics[width=0.7\textwidth, valign=c, page=7]{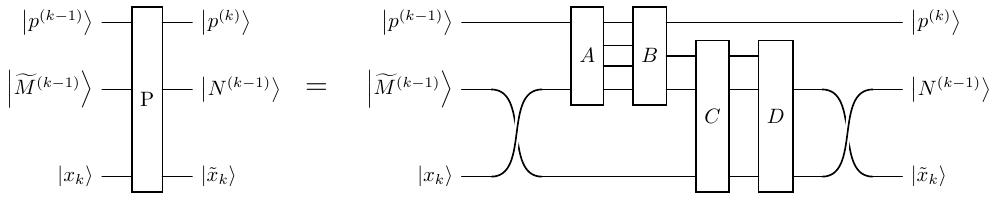}
    \caption{Preprocessing operation $P_k$ consists of four steps $A$, $B$, $C$ and $D$. 
    Steps $A$ and $B$ are needed to modify $x_k$ and update weight information $p^{(k-1)}$ according the newly arrived symbol $x_k$.
    Step $C$ modifies Gelfand--Tsetlin pattern $\widetilde{M}^{(k-1)}$ by adding one new row and shifting other rows according to the newly arrived symbol $x_k$.
    Finally, step $D$ uncomputes one auxilary register.}
    \label{fig:CG_circuit_P}
\end{figure*}

The transformation $A$ records the value of $x_k \in [d]$ and transforms it into $\tilde{x}_k \in [k]$: 
\begin{align}
    A &: \ket{p^{(k-1)}}\ket{x_k} \to \ket{p^{(k-1)}}\ket{x_k}\ket{c_k}\ket{\tilde{x}_k} \\
    c_k &:= \begin{cases}
        1 &\text{if } x_k \in p^{(k-1)} \\
        0 &\text{if } x_k \notin p^{(k-1)}
    \end{cases} \\
    \tilde{x}_k &:= \begin{cases}
        i \text{ s.t. } p^{(k-1)}_i = x_k &\text{if } x_k \in p^{(k-1)} \\
        i \text{ s.t. } p^{(k-1)}_{i-1} < x_k < p^{(k-1)}_{i} &\text{if } x_k \notin p^{(k-1)}
    \end{cases}
\end{align}
where $c_k \in \set{0,1}$ is a bit which indicated if the symbol $x_k$ is new or not (if $x_k \in p^{(k-1)}$ then the symbol $x_k$ have already appeared before), and $i$ is a position of $x_k$ within tuple $p^{(k-1)}$.
Note that $A$ is clearly reversible. 
The elementary gates in the expanded circuit of \cref{fig:CG_circuit_A} are reversible comparison gates, following the preprocessing construction of Ref.~\cite[Sec.~6]{burchardt2025krovi}.
We write unused slots of $p^{(k-1)}$ as zeros and ignore them in the comparisons.
Equivalently, the two values computed by $A$ are
\begin{align}
    c_k
    &=
    \sum_{i=1}^{k-1}
    \mathbf{1}\{p_i^{(k-1)}=x_k\},\\
    \tilde{x}_k
    &=
    1+
    \sum_{i=1}^{k-1}
    \mathbf{1}\{0<p_i^{(k-1)}<x_k\}.
\end{align}
Because the nonzero entries of $p^{(k-1)}$ are distinct, the first sum is a bit.
The gate $A_i$ acts on the old-symbol bit by
\begin{align}
    A_i:\quad
    \ket{p_i^{(k-1)}}\ket{x_k}\ket{b}
    \mapsto
    \ket{p_i^{(k-1)}}\ket{x_k}
    \ket{b\oplus \mathbf{1}\{p_i^{(k-1)}=x_k\}} .
\end{align}
The gate $A_i'$ acts on $r$ register by
\begin{align}
    A_i':\quad
    \ket{p_i^{(k-1)}}\ket{x_k}\ket{r}
    \mapsto
    \ket{p_i^{(k-1)}}\ket{x_k}
    \ket{r+\mathbf{1}\{0<p_i^{(k-1)}<x_k\}} .
\end{align}

\begin{figure*}[!h]
    \centering
    $\includegraphics[width=0.2\textwidth, valign=c, page=8]{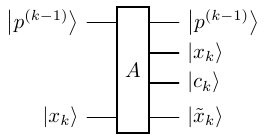}
    =
    \includegraphics[width=0.7\textwidth, valign=c, page=9]{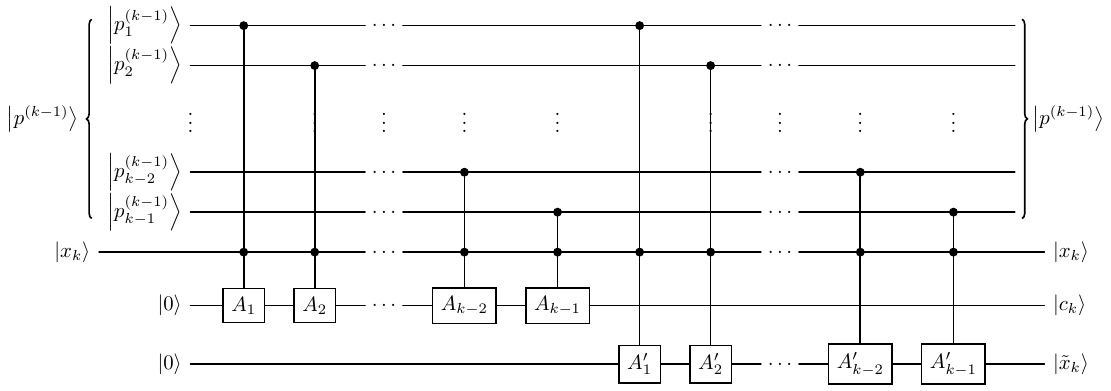}$
    \caption{
    Operation $A$.
    The expanded circuit computes $c_k$ with the equality tests $A_i$ and computes $\tilde{x}_k$ with the rank increments $A_i'$ defined above.}
    \label{fig:CG_circuit_A}
\end{figure*}

The operation $B$ updates the tuple $p^{(k-1)}$ to $p^{(k)}$ depending on the value of $c_k$:
\begin{align}
    B &: \ket{p^{(k-1)}}\ket{x_k}\ket{c_k}\ket{\tilde{x}_k} \to \ket{p^{(k)}}\ket{c_k}\ket{\tilde{x}_k} \\
    p^{(k)} &:= \begin{cases}
        (p^{(k-1)},0) &\text{if } c_k=1 \\
        (p^{(k-1)}_1,\dotsc,p^{(k-1)}_{\tilde{x}_k-1},x_k,p^{(k-1)}_{\tilde{x}_k},\dotsc,p^{(k-1)}_{k-1}) &\text{if } c_k=0
    \end{cases} 
\end{align}
The implementation of $B$ is a controlled insertion followed by a controlled erasure of the old copy of $x_k$.
Let $q=(q_1,\ldots,q_k)$ initially denote the list $(p_1^{(k-1)},\ldots,p_{k-1}^{(k-1)},x_k)$.
For $i=1,\ldots,k-1$, the gate $B_i$ is the adjacent controlled swap
\begin{align}
    B_i:\quad
    \ket{q_i}\ket{q_{i+1}}\ket{c_k}\ket{\tilde{x}_k}
    \mapsto
    \begin{cases}
        \ket{q_{i+1}}\ket{q_i}\ket{c_k}\ket{\tilde{x}_k},
        & c_k=0\ \mathrm{and}\ i\geq \tilde{x}_k,\\
        \ket{q_i}\ket{q_{i+1}}\ket{c_k}\ket{\tilde{x}_k},
        & \mathrm{otherwise}.
    \end{cases}
\end{align}
Applied in the order $B_{k-1}B_{k-2}\cdots B_1$, these gates move the new symbol from the last slot into position $\tilde{x}_k$ and shift all later slots one step to the right.
If $c_k=1$, no $B_i$ gate fires and the last slot still contains the duplicate value $x_k$.
For $i=1,\ldots,k-1$, the gate $B_i'$ erases this duplicate by a controlled modular subtraction on the last slot:
\begin{align}
    B_i':\quad
    \ket{p_i^{(k-1)}}\ket{q_k}\ket{c_k}\ket{\tilde{x}_k}
    \mapsto
    \begin{cases}
        \ket{p_i^{(k-1)}}\ket{q_k-p_i^{(k-1)}}\ket{c_k}\ket{\tilde{x}_k},
        & c_k=1\ \mathrm{and}\ \tilde{x}_k=i,\\
        \ket{p_i^{(k-1)}}\ket{q_k}\ket{c_k}\ket{\tilde{x}_k},
        & \mathrm{otherwise}.
    \end{cases}
\end{align}
On the promised inputs, the triggered subtraction has $q_k=x_k=p_i^{(k-1)}$, so the last slot becomes $0$.
Thus the $B_i$ gates handle the new-symbol case and the $B_i'$ gates handle the old-symbol case.

\begin{figure*}[!h]
    \centering
    $\includegraphics[width=0.2\textwidth, valign=c, page=10]{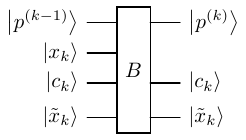}
    =
    \includegraphics[width=0.7\textwidth, valign=c, page=11]{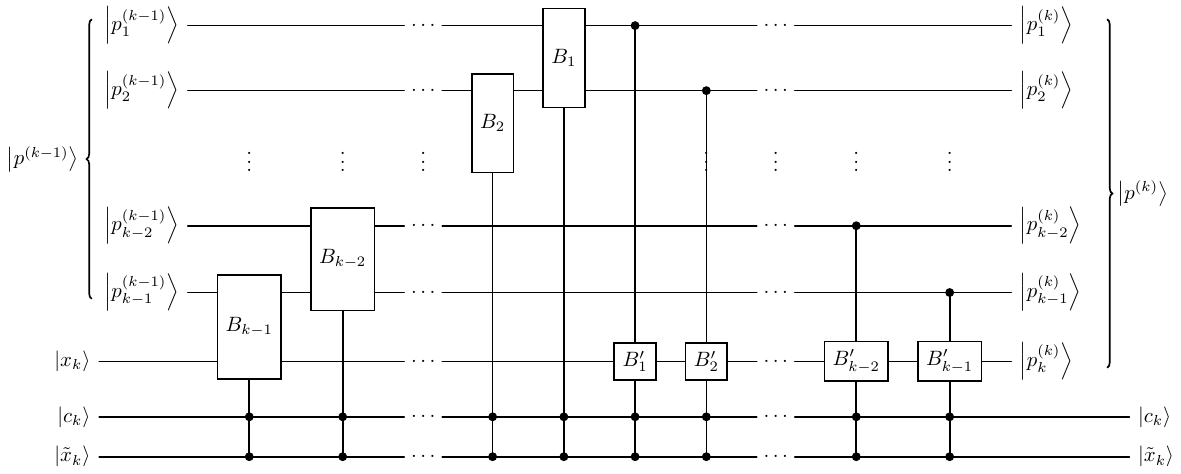}$
    \caption{
    Operation $B$.
    The expanded circuit uses the controlled swaps $B_i$ and the controlled erasures $B_i'$ defined above.}
    \label{fig:CG_circuit_B}
\end{figure*}

Next, the operation $C$ transforms the GT pattern $\ket{\widetilde{M}^{(k-1)}}$ differently according to the value of $c_k$:
\begin{align}
    C &: \ket{c_k}\ket{\tilde{x}_k}\ket{\widetilde{M}^{(k-1)}} \to \ket{c_k}\ket{\tilde{x}_k}\ket{N^{(k-1)}} \\
    &\text{If } c_k = 1: \\
    N^{(k-1)}_l &:= \begin{cases}
        \widetilde{M}^{(k-1)}_l, &\text{if } 1 \leq l \leq k-1  \\
        (\widetilde{M}^{(k-1)}_{k-1},0) &\text{if } l = k
    \end{cases} \\
    &\text{If } c_k = 0: \\
    N^{(k-1)}_l &:= \begin{cases}
        \widetilde{M}^{(k-1)}_l &\text{if } 1 \leq l < \tilde{x}_k  \\
        (\widetilde{M}^{(k-1)}_{l-1},0) &\text{if } \tilde{x}_k \leq l \leq k 
    \end{cases} 
\end{align}
The circuit realizes $C$ by row-selector gates $C_l$, one for each output row $N^{(k-1)}_l$.
The controls of $C_l$ are $c_k$ and $\tilde{x}_k$.
For $1\leq l<k$, its data inputs are the two adjacent candidate rows $\widetilde{M}^{(k-1)}_{l-1}$ and $\widetilde{M}^{(k-1)}_l$.
For $l=k$, the second candidate is absent, and the only possible output is the padded last row.
We use the boundary convention $\widetilde{M}^{(k-1)}_0=(0)$.
The output written by $C_l$ is defined by the following formula:
\begin{align}
    C_l:\quad
    N^{(k-1)}_l
    =
    \begin{cases}
        \widetilde{M}^{(k-1)}_l,
        & l<k\ \mathrm{and}\ (c_k=1\ \mathrm{or}\ l<\tilde{x}_k),\\[2pt]
        \bigl(\widetilde{M}^{(k-1)}_{l-1},0\bigr),
        & l=k\ \mathrm{or}\ (c_k=0\ \mathrm{and}\ \tilde{x}_k\leq l\leq k-1).
    \end{cases}
\end{align}
Thus $C_l$ either routes the old unshifted row $\widetilde{M}^{(k-1)}_l$ to output row $l$, or routes the previous row $\widetilde{M}^{(k-1)}_{l-1}$ with a trailing zero appended.
The gate $C$ is shown schematically in \cref{fig:CG_circuit_C}.
\begin{figure*}[!h]
    \centering
    $\includegraphics[width=0.18\textwidth, valign=c, page=12]{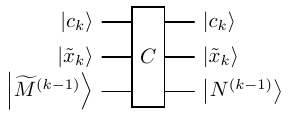}
    =
    \includegraphics[width=0.76\textwidth, valign=c, page=14]{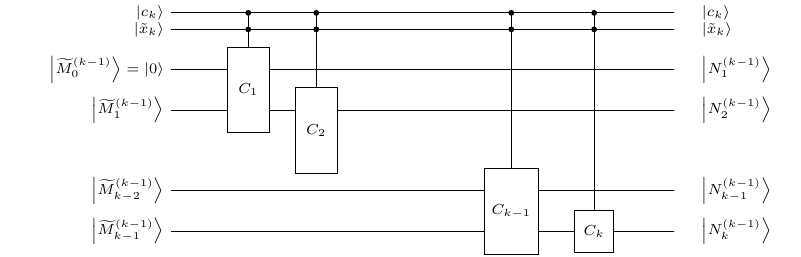}$
    \caption{Operation $C$ copies or shifts the compressed Gelfand--Tsetlin rows according to $c_k$ and $\tilde{x}_k$.
    The expanded circuit uses the row-selectors $C_l$ defined above.}
    \label{fig:CG_circuit_C}
\end{figure*}
The intuition behind transformation $C$ is as follows.
When $x_k$ is an old symbol, i.e. when $c_k=1$, no new alphabet row is inserted and operation $C$ appends the padded row $(\widetilde{M}^{(k-1)}_{k-1},0)$ only to keep the work register at fixed length $k$.
When $x_k$ is a new symbol, i.e. when $c_k=0$, operation $C$ inserts a zero-extension row at position $\tilde{x}_k$ and shifts every later compressed row down by one position.
The map is reversible because the registers $c_k$ and $\tilde{x}_k$ specify exactly whether the last padded row should be removed or whether the inserted zero-extension row at position $\tilde{x}_k$ should be removed.

Finally, the operation $D$ uncomputes the additional bit $c_k$:
\begin{align}
    D : \ket{c_k}\ket{\tilde{x}_k}\ket{N^{(k-1)}} \to \ket{0}\ket{\tilde{x}_k}\ket{N^{(k-1)}},
\end{align}
To that end, define
\begin{align}
\label{eq:appC_delta_def}
    \Delta_{\tilde{x}_k}(N)
    \coloneqq
    \sum_{i=1}^{\tilde{x}_k} N^{(k-1)}_{\tilde{x}_k,i}
    -
    \sum_{i=1}^{\tilde{x}_k-1} N^{(k-1)}_{\tilde{x}_k-1,i}.
\end{align}
After operation $C$, the old-symbol case $c_k=1$ is equivalent to $\Delta_{\tilde{x}_k}(N)>0$, while the new-symbol case $c_k=0$ is equivalent to $\Delta_{\tilde{x}_k}(N)=0$.
Thus $D$ is implemented by reversibly computing the comparison bit $b_k=\mathbf{1}\{\Delta_{\tilde{x}_k}(N)>0\}$, applying $c_k\mapsto c_k\oplus b_k$, and uncomputing all arithmetic work registers.
In the compact circuit, the gate $\mathrm{cmp}_{\Delta}$ denotes the reversible comparison subroutine shown below.
\begin{align}
    \mathrm{cmp}_{\Delta}:
    \ket{\tilde{x}_k}\ket{N^{(k-1)}}\ket{0}_{b_k}
    \mapsto
    \ket{\tilde{x}_k}\ket{N^{(k-1)}}\ket{\mathbf{1}\{\Delta_{\tilde{x}_k}(N)>0\}}_{b_k},
\end{align}
where any internal arithmetic scratch is suppressed in the compact drawing.
The inverse gate $\mathrm{cmp}_{\Delta}^{\dagger}$ uncomputes this comparison bit after the controlled flip on $c_k$.
The circuit displays one implementation of $\mathrm{cmp}_{\Delta}$ and its inverse.
The left half of the expanded circuit is $\mathrm{cmp}_{\Delta}$, and the right half is $\mathrm{cmp}_{\Delta}^{\dagger}$.
The row-sum gates are defined as follows.
\begin{align}
    S_+:
    \ket{\tilde{x}_k}\ket{N^{(k-1)}}\ket{0}_{s_+}
    &\mapsto
    \ket{\tilde{x}_k}\ket{N^{(k-1)}}\ket{\sum_i N^{(k-1)}_{\tilde{x}_k,i}}_{s_+},\\
    S_-:
    \ket{\tilde{x}_k}\ket{N^{(k-1)}}\ket{0}_{s_-}
    &\mapsto
    \ket{\tilde{x}_k}\ket{N^{(k-1)}}\ket{\sum_i N^{(k-1)}_{\tilde{x}_k-1,i}}_{s_-}.
\end{align}
The zero-comparison gate $\mathrm{cmp}_0$ is defined as follows.
\begin{align}
    \mathrm{cmp}_0:
    \ket{s_+}\ket{s_-}\ket{0}_{b_k}
    \mapsto
    \ket{s_+}\ket{s_-}\ket{\mathbf{1}\{s_+-s_->0\}}_{b_k}.
\end{align}
The gates $S_+^\dagger$, $S_-^\dagger$, and $\mathrm{cmp}_0^\dagger$ are the inverse arithmetic gates that erase $s_+$, $s_-$, and $b_k$, respectively.
The controlled-NOT from $b_k$ to $c_k$ implements $c_k\mapsto c_k\oplus b_k$.
Thus the nontrivial gates in the expanded implementation of $D$ are the row-sum gates $S_\pm$, the zero-comparison gate $\mathrm{cmp}_0$, their inverses, and this controlled-NOT.
The schematic gate $D$ and its expanded implementation are shown in \cref{fig:CG_circuit_D}.

\begin{figure*}[!h]
    \centering
    $\includegraphics[width=0.22\textwidth, valign=c, page=13]{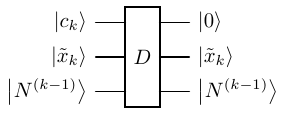}
    =
    \includegraphics[width=0.54\textwidth, valign=c, page=15]{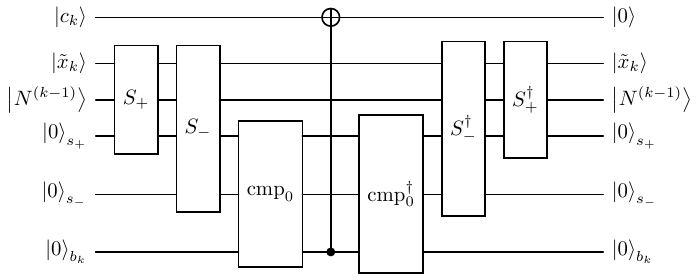}$
    \caption{
    Operation $D$ uncomputes $c_k$ by comparing the row-sum gap $\Delta_{\tilde{x}_k}(N)$ with zero.
    The left schematic gate is the operation $D$, and the right circuit expands it into row-sum, zero-comparison, and uncomputation gates.}
    \label{fig:CG_circuit_D}
\end{figure*}

\begin{figure*}[!h]
    \centering
    \includegraphics[width=\textwidth, page=6]{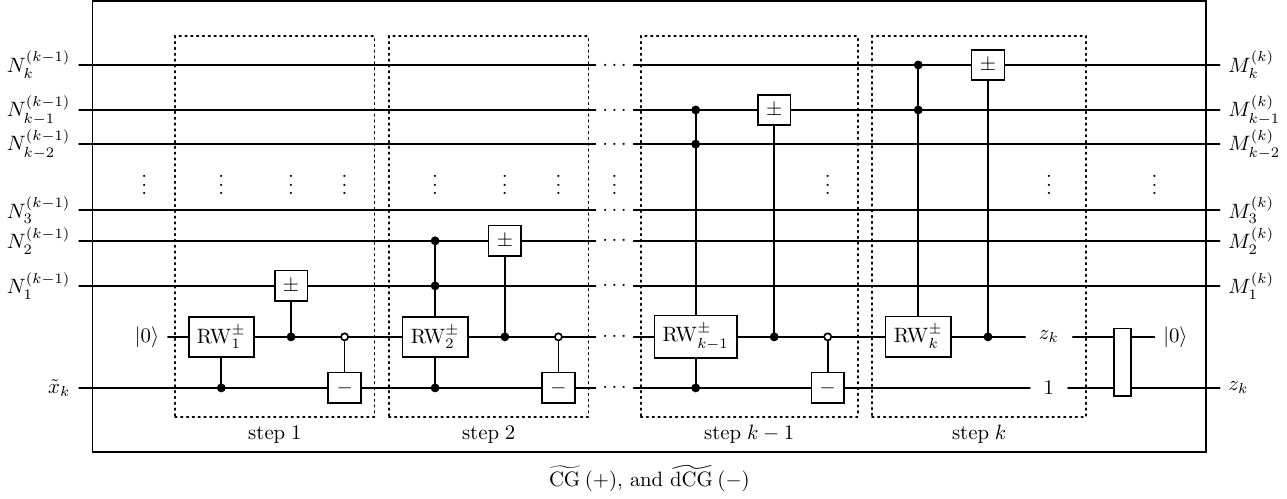}
    \caption{Full circuit for (dual) Clebsch--Gordan transforms $\widetilde{\mathrm{CG}}$ and $\widetilde{\mathrm{dCG}}$ transforms ($``+''$ corresponds to $\widetilde{\mathrm{CG}}$ and $``-''$ corresponds to $\widetilde{\mathrm{dCG}}$). 
    They consist of reduced Wigner transforms $\mathrm{RW}^{\pm}$ and simple arithmetic gates. 
    We refer to \cite{grinko2023gelfand} for implementation details of the gates.
    }
    \label{fig:CG_circuit}
\end{figure*}

\begin{remark}
    Note that implementing $\mathrm{dCG}^\dagger_k$ is easy if we have a promise, that the output irrep of $\mathrm{dCG}_k$ is described by a Young diagram, i.e. there are no negative entries in the highest weight corresponding to the output. 
    But this is indeed the case, since the input of the multiplicity space register from $\mathrm{CG}_k$ in \cref{fig:pr_O_U} ensures that the output of $\mathrm{dCG}^\dagger_k$ is a valid Young diagram: the ``minus'' gates in \cref{fig:CG_circuit} when run in reverse never produce negative entries.
\end{remark}

In total, operations $A,B,C,D$ can be done in $O(k^3)$ gate and depth complexity \cite{burchardt2025krovi}.
Depth and gate complexity of $\widetilde{\mathrm{CG}}$ and $\widetilde{\mathrm{dCG}}$ transforms is $\widetilde{O}(k^4)$, while memory complexity is $\widetilde{O}(k^2)$ \cite{nguyen2023mixed,grinko2023gelfand,burchardt2025krovi}. 
So, total complexity of $\mathrm{CG}_k$ and $\mathrm{dCG}_k$ is $\widetilde{O}(k^4)$, which implies $\widetilde{O}(k^4)$ gate and depth complexity for our compressed oracle $\fO$.

\subsection{\texorpdfstring{On implementation of mixed high-dimensional Clebsch--Gordan transforms}{On implementation of mixed high-dimensional Clebsch--Gordan transforms}}
\label{app:mixed_cg_transforms}

For mixed sequences of defining and dual defining factors, polynomial Gelfand--Tsetlin patterns are replaced by rational Gelfand--Tsetlin patterns.
The sparse-alphabet idea still applies, but the criterion for keeping a physical alphabet positions must be changed.
Positive and negative ``branching'' can cancel in the signed weight.
The support register must therefore record this cancellation.

Write a rational Gelfand--Tsetlin pattern as
\begin{align}
    M
    =
    (M_1\sqsubseteq M_2\sqsubseteq\cdots\sqsubseteq M_d),
    \qquad
    M_r=(M_{r,1}\geq\cdots\geq M_{r,r})\in\mathbb Z^r,
\end{align}
with $M_0=\varnothing$.  Each row has positive and negative parts,
\begin{align}
    M_r
    =
    (\alpha^{(r)}_1,\ldots,\alpha^{(r)}_a,
    0,\ldots,0,
    -\beta^{(r)}_b,\ldots,-\beta^{(r)}_1),
\end{align}
where $\alpha^{(r)}$ and $\beta^{(r)}$ are partitions, $a$ and $b$ may depend on $r$, and we take $\alpha^{(0)}=\beta^{(0)}=\varnothing$.
Interlacing implies
\begin{align}
    \alpha^{(r-1)}\sqsubseteq\alpha^{(r)},
    \qquad
    \beta^{(r-1)}\sqsubseteq\beta^{(r)}.
\end{align}
Thus the two nonnegative branching contents are
\begin{align}
    \mu_r^+
    &\coloneqq
    |\alpha^{(r)}|-|\alpha^{(r-1)}|,
    &
    \mu_r^-
    &\coloneqq
    |\beta^{(r)}|-|\beta^{(r-1)}|.
\end{align}
The signed row weight introduced before and the unsigned row activity are
\begin{align}
\label{eq:mixed_activity_def}
    w_r(M)
    &=
    \mu_r^+-\mu_r^-,
    \\
    \operatorname{act}_r(M)
    &\coloneqq
    \mu_r^++\mu_r^-
    =
    \sum_{j=1}^{r}|M_{r,j}|
    -
    \sum_{j=1}^{r-1}|M_{r-1,j}|.
\end{align}
Both $\mu_r^+$ and $\mu_r^-$ are nonnegative.
Hence $w_r=0$ may result either from $\mu_r^+=\mu_r^-=0$ or from a nontrivial cancellation $\mu_r^+=\mu_r^->0$, whereas $\operatorname{act}_r(M)=0$ detects the first case.

For a row $\nu$ of length $r-1$, define its canonical zero extension by
\begin{align}
\label{eq:mixed_zero_extension}
    \iota_0(\nu)
    \coloneqq
    \operatorname{sort}_{\downarrow}(\nu_1,\ldots,\nu_{r-1},0),
\end{align}
where $\operatorname{sort}_{\downarrow}$ sorts the entries in weakly decreasing order.
Equivalently, the new zero is inserted at the wall between the positive and negative entries.

\begin{lemma}[Zero-activity criterion]
\label{lem:mixed_zero_activity}
For every rational Gelfand--Tsetlin pattern and every $r\in[d]$,
\begin{align}
\label{eq:mixed_activity_zero_criterion}
    \operatorname{act}_r(M)=0
    \quad\Longleftrightarrow\quad
    M_r=\iota_0(M_{r-1}).
\end{align}
\end{lemma}
\begin{proof}
Interlacing gives the two inclusions of positive and negative diagrams above, so $\mu_r^+,\mu_r^-\geq0$.
If $\operatorname{act}_r(M)=0$, both terms $\mu_r^+,\mu_r^-$ vanish.
The positive and negative nonzero entries of $M_r$ therefore coincide with those of $M_{r-1}$.
Since $M_r$ has one additional coordinate, that coordinate must be zero, and weakly decreasing order places it exactly as in $\iota_0(M_{r-1})$.
The converse is immediate because canonical zero extension leaves both diagrams unchanged.
\end{proof}

The sparse alphabet is therefore the increasing list associated with
\begin{align}
\label{eq:mixed_activity_support}
    p(M)
    \coloneqq
    \{r\in[d]:\operatorname{act}_r(M)>0\},
\end{align}
not the set of indices for which $w_r\neq0$.  
If the top row is $[\alpha,\beta]$, the activity is
\begin{align}
\label{eq:mixed_activity_telescopes}
    \sum_{r=1}^{d}\operatorname{act}_r(M)
    =
    |\alpha|+|\beta|.
\end{align}
Indeed, the positive increments sum to $|\alpha|$, and the negative increments sum to $|\beta|$.
If $k_+$ defining and $k_-$ dual defining factors have been used, rational Pieri branching gives $|\alpha|=k_+-s$ and $|\beta|=k_--s$ for some contraction number $s$~\cite{nguyen2023mixed,grinko2023gelfand}.
Consequently, after at most $k$ total factors, the activity support remains small:
\begin{align}
    p(M)
    \leq
    \sum_{r=1}^{d}\operatorname{act}_r(M)
    =
    |\alpha|+|\beta|
    \leq
    k_+ + k_-
    \leq k.
\end{align}
So, the alphabet positions still use only $O(k\log d)$ bits, and the retained rational pattern has at most $k$ nontrivial branching rows.

Steps $A$ and $B$ are unchanged after replacing $p^{(k-1)}$ by the activity support in Eq.~\eqref{eq:mixed_activity_support}.
They compute whether the incoming physical label $x_k$ is already active and update the sorted support list.
The only changes are in preprocessing steps $C$ and $D$.
First, every polynomial padding operation $(\nu,0)$ in step $C$ must be replaced by the insertion $\iota_0(\nu)$. 
Thus the mixed row-copying rule is
\begin{align}
\label{eq:mixed_C_gate}
    c_k=1:
    \quad
    N_l^{(k-1)}
    &=
    \begin{cases}
    \widetilde M_l^{(k-1)},&1\leq l\leq k-1,\\
    \iota_0\!\left(\widetilde M_{k-1}^{(k-1)}\right),&l=k,
    \end{cases}
    \\
    c_k=0:
    \quad
    N_l^{(k-1)}
    &=
    \begin{cases}
    \widetilde M_l^{(k-1)},&1\leq l<\tilde x_k,\\
    \iota_0\!\left(\widetilde M_{l-1}^{(k-1)}\right),&\tilde x_k\leq l\leq k.
    \end{cases}
\end{align}
We keep the convention
\begin{align}
    c_k=1
    \quad\Longleftrightarrow\quad
    x_k\in p^{(k-1)}.
\end{align}
After the row-copying step, we define
\begin{align}
\label{eq:mixed_activity_gap}
    \Delta^{\mathrm{mix}}_{\tilde x_k}(N)
    \coloneqq
    \sum_{i=1}^{\tilde x_k}
    \left|N^{(k-1)}_{\tilde x_k,i}\right|
    -
    \sum_{i=1}^{\tilde x_k-1}
    \left|N^{(k-1)}_{\tilde x_k-1,i}\right|.
\end{align}
On valid preprocessing inputs, this is precisely the activity of the branching edge at $\tilde x_k$.  
Therefore,
\begin{align}
    c_k=1
    &\quad\Longleftrightarrow\quad
    \Delta^{\mathrm{mix}}_{\tilde x_k}(N)>0,
    &
    c_k=0
    &\quad\Longleftrightarrow\quad
    \Delta^{\mathrm{mix}}_{\tilde x_k}(N)=0.
\end{align}
The reversible uncomputation of $c_k$ is therefore the gate
\begin{align}
\label{eq:mixed_D_gate}
    D_{\mathrm{mix}}:
    \ket{c_k}\ket{\tilde x_k}\ket{N}
    \longmapsto
    \ket{c_k\oplus\mathbf 1\!\left[
    \Delta^{\mathrm{mix}}_{\tilde x_k}(N)>0
    \right]}
    \ket{\tilde x_k}\ket{N}.
\end{align}
On valid preprocessing inputs this sends $c_k$ to $0$.
An implementation can avoid absolute-value arithmetic by testing the equivalent predicate
\begin{align}
\label{eq:mixed_D_comparator}
    \mathbf 1\!\left[
    \Delta^{\mathrm{mix}}_{\tilde x_k}(N)>0
    \right]
    =
    \mathbf 1\!\left[
    N^{(k-1)}_{\tilde x_k}
    \neq
    \iota_0\!\left(N^{(k-1)}_{\tilde x_k-1}\right)
    \right].
\end{align}
On valid preprocessing inputs this predicate is equal to $c_k$.
Computing this predicate into scratch ancilla register, flipping $c_k$, and uncomputing the scratch ancilla realizes the same gate.

The distinction between signed weight and activity is already visible in the adjoint highest weight $(1,0,\ldots,0,-1)$.  For each $s\in\{2,\ldots,d\}$, consider the rational pattern
\begin{align}
\label{eq:mixed_adjoint_pattern}
    M_j
    &=
    (0,\ldots,0),
    &&j<s,
    \\
    M_j
    &=
    (1,0,\ldots,0,-1),
    &&j\geq s.
\end{align}
These are the $d-1$ zero-weight Gelfand--Tsetlin basis states of this irrep.
Every signed component vanishes, but
\begin{align}
    \operatorname{act}_s(M)=2,
    \qquad
    \operatorname{act}_j(M)=0\quad(j\neq s),
    \qquad
    p(M)=\{s\}.
\end{align}
Thus the location of the nontrivial zero-weight branching edge stores the missing $\log(d-1)$ qubits.
A support rule based on $w_j\neq0$ would erase this information.

\section{\texorpdfstring{$\mathfrak{S}_d$ oracles}{Symmetric group oracles}}
\label{sec:permutation_fourier_database}

For $G=\mathfrak{S}_d$ we use the defining permutation representation $V_\pi\ket{x}=\ket{\pi(x)}$ on $\mathbb{C}^d$.
The Fourier memory can be obtained from the ordinary group algebra by the quantum Fourier transform
\begin{align}
\label{eq:perm_qft}
    \mathrm{QFT}_{\mathfrak{S}_d}\ket{\pi}
    =
    \bigoplus_{\lambda\vdash d}
    \sqrt{\frac{d_\lambda}{d!}}
    \sum_{P,Q}
    \rho_\lambda(\pi)_{P,Q}
    \ket{\lambda,P,\bar Q}.
\end{align}
In this basis, one query is the Clebsch--Gordan update for tensoring an $\mathfrak{S}_d$ irrep with the defining representation.
Writing the corresponding coefficients as $\CG^{\lambda,R}_{\mu,P,x}$ and $\CG^{\mu,Q}_{\lambda,R',y}$,
\begin{align}
\label{eq:perm_fourier_oracle_action}
    \fO\ket{x}\ket{\mu,P,\bar Q}
    =
    \sum_{y=1}^{d}
    \ket{y}
    \sum_{\lambda,R,R'}
    \CG^{\lambda,R}_{\mu,P,x}
    \overline{\CG^{\mu,Q}_{\lambda,R',y}}\,
    \ket{\lambda,R,\bar R'}.
\end{align}
Multiplicity labels, if present, are included in $R$ and $R'$.
This is the specialization of the general Fourier oracle to the permutation representation.

The standard database object for permutations is a partial injection $I:A\hookrightarrow B$, with $A,B\subseteq[d]$.
Its consistent-permutation state in the group algebra is
\begin{align}
\label{eq:perm_omega_I}
    \ket{\Omega_I}
    \coloneqq
    \frac{1}{\sqrt{(d-|I|)!}}
    \sum_{\pi:\,\pi|_A=I}
    \ket{\pi}.
\end{align}
The exact oracle update on these states is the familiar partial-injection rule
\begin{align}
\label{eq:perm_partial_update}
    \fO_{\Omega}\ket{x}\ket{\Omega_I}
    =
    \begin{cases}
    \ket{I(x)}\ket{\Omega_I},
    & x\in\operatorname{dom}(I),\\[2mm]
    \displaystyle
    \frac{1}{\sqrt{d-|I|}}
    \sum_{y\notin\operatorname{Im}(I)}
    \ket{y}\ket{\Omega_{I\cup\{x\mapsto y\}}},
    & x\notin\operatorname{dom}(I).
    \end{cases}
\end{align}
Applying Eq.~\eqref{eq:perm_qft} gives the Fourier coordinates of this database state,
\begin{align}
\label{eq:perm_omega_fourier}
    \ket{\widehat{\Omega}_I}
    \coloneqq
    \mathrm{QFT}_{\mathfrak{S}_d}\ket{\Omega_I}
    =
    \bigoplus_{\lambda\vdash d}
    \sqrt{\frac{d_\lambda}{d!(d-|I|)!}}
    \sum_{P,Q}
    \sum_{\pi:\,\pi|_A=I}
    \rho_\lambda(\pi)_{P,Q}
    \ket{\lambda,P,\bar Q}.
\end{align}
Thus the partial-injection update and the Fourier Clebsch--Gordan update are the same exact oracle written in two bases.
The states $\ket{\Omega_I}$ are useful because the update rule is local in the partial injection.
The Fourier states $\ket{\lambda,P,\bar Q}$ are useful because exactness follows directly from Peter--Weyl orthogonality and Clebsch--Gordan unitarity.

\subsection{\texorpdfstring{Comparison with Carolan's compressed permutation oracle}{Comparison with Carolan's compressed permutation oracle}}
\label{sec:comparison_carolan_permutation_oracle}

Carolan's oracle~\cite{carolan2025compressed} is defined on an orthonormal database basis $\{\ket{I}\}_I$ indexed by partial injections $I$.
In the notation used above, for every $x\notin\operatorname{dom}(I)$ the compression operation $\mathrm{pC}_x$ swaps
\begin{align}
\label{eq:carolan_compression_swap}
    \ket{I}
    \quad\longleftrightarrow\quad
    \frac{1}{\sqrt{d-|I|}}
    \sum_{y\notin\operatorname{Im}(I)}
    \ket{I\cup\{x\mapsto y\}},
\end{align}
and acts as the identity on the orthogonal complement of the direct sum of these two-dimensional subspaces.
The controlled compression is
\begin{align}
    \mathrm{pC}\ket{x}\ket{I}
    =
    \ket{x}\,\mathrm{pC}_x\ket{I}.
\end{align}
The partial-table query acts on the standard permutation-oracle registers as
\begin{align}
\label{eq:carolan_partial_query}
    P\ket{x,z}\ket{I}
    =
    \begin{cases}
    \ket{x,z\oplus_d I(x)}\ket{I},
    &x\in\operatorname{dom}(I),\\
    \ket{x,z}\ket{I},
    &x\notin\operatorname{dom}(I),
    \end{cases}
\end{align}
while inversion of the database is the unitary
\begin{align}
\label{eq:carolan_flip}
    F\ket{I}=\ket{I^{-1}}.
\end{align}
With the database initialized in the empty state, Carolan's forward and inverse compressed queries are therefore
\begin{align}
\label{eq:carolan_oracles}
    \mathrm{pC}\,P\,\mathrm{pC}^{\dagger}
    \qquad\text{and}\qquad
    F\,\mathrm{pC}\,P\,\mathrm{pC}^{\dagger}\,F^{\dagger},
\end{align}
respectively.
Both $\mathrm{pC}$ and $F$ are involutions, so the daggers may be suppressed.
The second formula first reverses the partial injection, applies the same forward compressed query, and reverses it again.

The mathematical difference from our oracle is the following.
The state $\ket{\Omega_I}$ in Eq.~\eqref{eq:perm_omega_I} is a normalized superposition of all full permutations extending $I$; different compatible $I$ give nonorthogonal states.
On that nonorthogonal family, Eq.~\eqref{eq:perm_partial_update} is an exact identity for the purified random permutation.
Carolan instead declares the records $\ket I$ orthogonal and encodes the required correlations by conjugating the partial-table query with $\mathrm{pC}$.
Consequently, replacing $\ket{\Omega_I}$ by $\ket I$ in Eq.~\eqref{eq:perm_partial_update} does not give Eq.~\eqref{eq:carolan_oracles}.
The two compression layers create interference between an undefined entry and the uniform superposition of its collision-free extensions, and the database update need not be append-only.

For tensor powers of the defining permutation representation of $\mathfrak S_d$, the Schur--Weyl commutant is the image of the partition algebra $P_t(d)$, with a kernel outside the stable range~\cite{halverson2005partition,benkart2017dimensions,bowman2022integral}.
Consequently, the Schur transform that block diagonalizes $(\mathbb C^d)^{\otimes t}$ for this representation is a partition-algebra Schur transform: it separates irreducible $\mathfrak S_d$ modules from the corresponding partition-algebra multiplicity spaces.
This should be distinguished from the ordinary group-algebra Peter--Weyl basis of $\mathbb C[\mathfrak S_d]$ used by our Fourier oracle.
In other words, if one reaches the visible $t$-query subspace through tensor powers of the defining representation, the relevant Schur transform is governed by the partition algebra, whereas the Fourier memory itself is expressed in the Peter--Weyl basis exposed by $\mathrm{QFT}_{\mathfrak S_d}$.
The two constructions therefore optimize different properties.
The Fourier oracle in Eq.~\eqref{eq:perm_fourier_oracle_action}, equivalently the consistent-state oracle in Eq.~\eqref{eq:perm_partial_update}, is exact and canonical from Peter--Weyl theory, but its database records are nonorthogonal.
In contrast, Carolan's oracle has orthogonal partial-injection records.
Its relation to the truly random bidirectional permutation oracle is instead approximate: Ref.~\cite[Theorem~5.19]{carolan2025compressed} bounds the distinguishing advantage of a $t$-query algorithm by $O(t^3/d^{1/4})$.

\subsection{\texorpdfstring{Equivalence of the oracle accesses to $U_g, U_{g^{-1}}$ and those to $V_g, V_{g^{-1}}$}{Equivalence of oracle accesses}}
\label{sec:equivalence_permutation_oracle}
The oracle accesses to $U_g, U_{g^{-1}}$ defined in Eq.~\eqref{eq:def_Ug} and those to $V_g, V_{g^{-1}}$ defined in Eq.~\eqref{eq:def_Vg} are equivalent since either can be simulated using two queries to the other.
The simulation can be done as follows:
\begin{align}
    U_g\ket{x,y}
    &= (V_{g^{-1}} \otimes I_d)\mathrm{CX}_d(V_g \otimes I_d)\ket{x,y},\\
    U_{g^{-1}}(I_d\otimes N_d)\mathrm{SWAP}\,U_g\ket{x,0}
    &= V_g\ket{x}\ket{0}.
\end{align}
Thus the second register in the second line is deterministically uncomputed to $\ket{0}$ and can be discarded.
Here $I_d$ is the identity operator on $\mathbb{C}^d$, $\mathrm{CX}_d$ and $\mathrm{SWAP}$ are two-qudit unitaries, and $N_d$ is the one-qudit modular-negation unitary defined by
\begin{align}
    \mathrm{CX}_d
    &\coloneqq
    \sum_{x,y\in\mathbb Z_d}\ketbra{x,x\oplus_d y}{x,y},
    \\
    \mathrm{SWAP}
    &\coloneqq
    \sum_{x,y\in\mathbb Z_d}\ketbra{y,x}{x,y},
    \\
    N_d
    &\coloneqq
    \sum_{z\in\mathbb Z_d}\ketbra{-z}{z}.
\end{align}

\end{document}